\documentclass[11pt]{article}
\usepackage[margin=1in]{geometry}
\pdfoutput=1
\usepackage[utf8]{inputenc}
\usepackage[english]{babel}
\usepackage[T1]{fontenc}
\usepackage{amsthm}
\usepackage{amsmath,amssymb,dsfont}
\usepackage{graphicx}       
\usepackage{physics}
\usepackage{hyperref}
\usepackage{relsize}
\usepackage{tcolorbox}
\usepackage{thm-restate}
\hypersetup{
    colorlinks=true,
    linkcolor=blue,
    citecolor=magenta,
    filecolor=magenta,      
    urlcolor=blue,
    pdftitle={},
    pdfpagemode=FullScreen,
    }
\usepackage{soul}
\usepackage{mathtools}
\usepackage{enumerate}
\usepackage[normalem]{ulem}
\usepackage{orcidlink}

\usepackage{graphicx,color,colordvi}
\usepackage{tikz}
\usepackage{ifthen}
\usepackage{pgfplots}
\pgfplotsset{compat=1.9}
\usetikzlibrary{shapes,arrows.meta}
\usetikzlibrary{positioning}
\usetikzlibrary{shapes.geometric}

\usepackage{tikz}
\usetikzlibrary{shapes}
\usepackage{float}
\makeatletter
\newlength\min@xx

\makeatother

\newtheorem{theorem}{Theorem}
\newtheorem{lemma}{Lemma}
\newtheorem{corollary}{Corollary}
\newtheorem{definition}{Definition}
\newtheorem{remark}{Remark}
\newtheorem{example}{Example}
\newtheorem{proposition}{Proposition}

\newcommand{\ZZ}{\mathbb{Z}}
\newcommand{\eps}{\varepsilon}
\DeclareMathOperator{\diam}{diam}

\newcommand{\heff}[2]{h_{\mathrm{eff,}#2}^{#1}}
\newcommand{\hefftot}[1]{h_{\mathrm{eff}}^{#1}}
\newcommand{\heffexact}[1]{h_{\mathrm{eff \: exact}}^{#1}}
\DeclareMathOperator{\FMED}{F_{\textup{MED}}}
\DeclareMathOperator{\F}{F}
\DeclareMathOperator{\im}{im}
\DeclareMathOperator{\poly}{poly}
\newcommand{\id}{\mathds 1}

\newcommand{\ellshield}{\ell_S}

\definecolor{Turquoise}{HTML}{14C7DE}
\definecolor{SkyBlue}{HTML}{3498DB}
\definecolor{CoralDark}{HTML}{FF6F61} 
\definecolor{purplemod}{HTML}{AF82F3}
\definecolor{darkgreen}{RGB}{106, 134, 104}
\definecolor{darkgreen2}{RGB}{90, 168, 143}

\def\Block[#1,#2,#3,#4]{
\def\r{0.3};
\ifthenelse{\NOT #4=0}{
\fill [#2] (-0.5,-0.5) rectangle ({#1-0.5},0.5);
}
\foreach \n in {1,...,#1}{ 
\shade[shading=ball, ball color=darkblue] ({\n-1},0) circle (\r);
}

\begin{scope}[decoration={brace,mirror,amplitude=7}]
\ifthenelse{#4=1}{
\draw [decorate] (-0.5,-0.6) --node[below=3mm]{$#3$} ({#1-0.5},-0.6);
}
\ifthenelse{#4=2}{
\draw [decorate] ({#1-0.5},0.6) --node[above=3mm]{$#3$} (-0.5,0.6);
}
\end{scope}
}

\usetikzlibrary{shapes,arrows,patterns,positioning,shapes.geometric,arrows.meta,decorations.markings}

\newcommand{\boxElement}[7]{
    \draw[#5, rounded corners=10pt, thick] (#1,#2) rectangle (#1 + #3,#2 + #4);
    \fill[#5, rounded corners=10pt, opacity=0.2] (#1,#2) rectangle (#1 + #3,#2 + #4);
    
    \node[fill=white, rounded corners=3pt, anchor=south west, inner sep=2pt] at (#1 + 0.25,#2 + #4 - 0.25) {\textcolor{#5}{\textbf{#6}}};

    \node[anchor=north west, text width={#3cm - 0.5cm}, align=left] at (#1 + 0.25, #2 + #4 - 0.25) {#7\hfill~};
    
}

\tikzset{
    plus arrow/.style={
        postaction={decorate},
        decoration={
            markings,
            mark=at position 1 with {
                \node[draw, circle, thick, fill=white, inner sep=0pt, minimum size=8pt] {\textbf{+}};
            }
        }
    }
}

\usepackage{authblk}
\author[1]{Samuel O. Scalet\orcidlink{0000-0002-9770-3137}\footnote{samuel.scalet@tum.de}}
\affil[1]{Department of Applied Mathematics and Theoretical Physics, University of Cambridge, United Kingdom}

 \author[1,2]{Angela Capel\orcidlink{0000-0001-6713-6760}}

\affil[2]{Fachbereich Mathematik, Universität Tübingen, 72076 Tübingen, Germany}
\author[3]{Anirban N. Chowdhury\orcidlink{0000-0001-9743-7978}}
 \affil[3]{IBM Quantum, IBM T.\ J.\ Watson Research Center, Yorktown Heights, NY, 10598, USA}
\author[1]{Hamza Fawzi\orcidlink{0000-0001-6026-4102}}
\author[4]{Omar~Fawzi\orcidlink{0000-0001-8491-0359}}
\affil[4]{Inria, ENS Lyon, UCBL, LIP, F-69342 Lyon Cedex 07, France}

\author[5]{Isaac H. Kim\orcidlink{0000-0001-7689-3157}}
\affil[5]{Department of Computer Science, UC Davis, Davis, CA, 95616, USA}
\author[6]{Arkin Tikku\orcidlink{0000-0002-5942-9325}}
\affil[6]{Centre for Engineered Quantum Systems, University of Sydney, Sydney, NSW 2006, Australia.}

\title{Classical Estimation of the Free Energy and Quantum Gibbs Sampling from the Markov Entropy Decomposition}

\date{}

\begin{document}

\maketitle

\begin{abstract}
We revisit the Markov Entropy Decomposition, a classical convex relaxation algorithm introduced by Poulin and Hastings to approximate the free energy in quantum spin lattices. We identify a sufficient condition for its convergence, namely the decay of the effective interaction.
The effective interaction, also known as Hamiltonians of mean force, is a widely established correlation measure, and we show our decay condition in 1D at any temperature as well as in the high-temperature regime under a certain commutativity condition on the Hamiltonian building on existing results.
This yields  polynomial and quasi-polynomial time approximation algorithms in these settings, respectively. Furthermore, the decay of the effective interaction implies the decay of the conditional mutual information for the Gibbs state of the system. We then use this fact to devise a rounding scheme that maps the solution of the convex relaxation to a global state and show that the scheme can be efficiently implemented on a quantum computer, thus proving efficiency of quantum Gibbs sampling under our assumption of decay of the effective interaction.
\end{abstract}

\section{Introduction}
\label{sec:introduction}

Determining thermal equilibrium properties of interacting quantum systems is a foundational yet challenging problem in many-body physics. Computing the free energy, or sampling from the equilibrium Gibbs distribution of a quantum Hamiltonian are known to be computationally hard in the worst-case. Therefore, it is important to identify instances where these problems can be solved efficiently and indeed design new algorithmic techniques for doing so. Recent years have seen a number of advances in this area, which have built upon deep insights into the structural properties of thermal states \cite{alhambra2023quantum}.

One such technique, known as the Markov Entropy Decomposition (MED), was proposed by Poulin and Hastings in \cite{Poulin2011}.
This algorithm extends ideas from convex optimization for bounding from below ground state energies \cite{gharibian2019almost,bravyi2019approximation,parekh2021application,hastings2022optimizing} to the setting of thermal states, thereby naturally providing \emph{lower bounds} to the free energy.
The MED has been studied numerically \cite{Ferris2013} but there are no known guarantees on its convergence.
The authors in \cite{Poulin2011} conjecture that its convergence is connected to the approximate \emph{saturation of strong subadditivity} also known as the decay of the \emph{conditional mutual information (CMI)}.

In this paper, we provide a sufficient condition for efficient convergence of the MED - that of an exponentially decaying effective interaction. 
The effective interaction is the difference between the parent Hamiltonian of marginals of the thermal state (the effective Hamiltonian) and the original Hamiltonian restricted to the support.
The effective Hamiltonian, which is also known as the Hamiltonian of mean force, has been widely studied in statistical and quantum many-body physics \cite{alhambra2023quantum,Talkner2020,Campisi2009,Kuwahara2024,Bilgin2010bp,Cresser2021,Hakoshima_2024,Trushechkin2022,bluhm2024}. Using known results, the exponential decay of the effective interaction with distance can be shown to hold in one-dimensional systems \cite{Kuwahara2024} and for higher-dimensional systems above a threshold temperature under an additional commutativity condition \cite{bluhm2024}. We thus prove that the MED gives efficient algorithms in these settings.

It is also related to the entanglement Hamiltonian which has been studied in the context of entanglement properties of ground states \cite{Rottoli_2025,Eisler_2022,Peschel_2009,Parisen2018,Eisler_2020,Kokail2021,Zache2022}. We carefully formalize the notion of an exponentially-decaying effective interaction in a way that is consistent with known rigorous, conjectured and experimental results, and show that this implies the decay of CMI. 1D systems of commuting Hamiltonians fulfill this property even exactly for a finite length scale. As mentioned earlier, it is known to hold for 1D Hamiltonians and for certain commuting systems in higher dimensions at high temperatures. It remains an open question how to generalize this property to noncommuting Hamiltonians beyond 1D since a claimed proof in \cite{Kuwahara2020} was found to be flawed. Our work can be seen as rigorously connecting the efficiency of algorithms for free energies (and thereby expectation values of local observables) to an established physical and information theoretic criterion.

Our work also highlights the use of \emph{convex programming relaxations} in designing provably efficient algorithms for free energy problems. Convex and semi-definite relaxations have been widely used for ground-state energy problems \cite{gharibian2019almost,bravyi2019approximation,parekh2021application,hastings2022optimizing}. Recent work has also extended its use to free energies through ideas similar to the MED \cite{Bravyi2022}; see also \cite{risteski2016partition} for related classical work. However, these approaches typically give algorithms that obtain constant-factor multiplicative approximations to the objective in polynomial time.
Our work establishes significantly stronger error guarantees for the convex relaxation considered here\,---\,in the two known cases where the exponential decay of the effective interaction holds, the MED has runtime that is (quasi-)polynomial in the inverse error.

\begin{figure}[h!]
    \centering
    \begin{tikzpicture}[font=\scriptsize, scale=.9]
        \def\toolcolor{darkgreen2}
        \def\assumptioncolor{CoralDark}
        \def\resultcolor{SkyBlue}

        \boxElement{0}{3}{6}{1.2}{\assumptioncolor}{Effective interaction}{Quasi-local decomposition of $\log(\rho_A)$,\\
        see Definition~\ref{def:heff}.};
        
        \boxElement{0}{0}{6}{1.2}{\toolcolor}{Decay CMI}{\phantom{asddd}$I(A:C|B)\le e^{-\mathcal{O}(d(A,C))}$};

        \boxElement{8}{3}{6}{1.2}{\toolcolor}{MED}{Classical (quasi-)polynomial time\\algorithm for $F$ and $\rho_A$.};

        \boxElement{8}{0}{6}{1.3}{\resultcolor}{Efficient Gibbs sampling}{(Quasi-)polynomial time quantum \\algorithm to prepare $\rho$.};

        \draw[black, thick, ->] (3.5, 3) -- (3.5, 1.2) node[midway,left] {Lemma \ref{lem:CMI_decay}} ;
        \draw[black, thick, ->] (6, 3.6) -- (8, 3.6) node[midway,above] {Corollary \ref{cor:free_energy_algorithm}};
        \draw[black, thick] (5.9, 1.1) -- (7.3, 1.7);
        \draw[black, thick] (8.1, 3.1) -- (7.3, 1.7);
        \draw[black, thick, ->] (7.3, 1.7) -- (8.05, 1.1) node[above,below left] {Corollary \ref{cor:circuitDecomposed}};
    \end{tikzpicture}
    \caption{Overview of the connections proven in this work. We prove a decay condition on the effective interaction and use it for the design of classical and quantum algorithms for computing the free energy $F$. Here, $\rho$ and $\rho_A$ denote the Gibbs state and its marginal, respectively, and $I(A:C|B)=S(AB)+S(BC)-S(ABC)-S(B)$ is the von Neumann conditional mutual information.}
    \label{fig:paper-structure}
\end{figure}
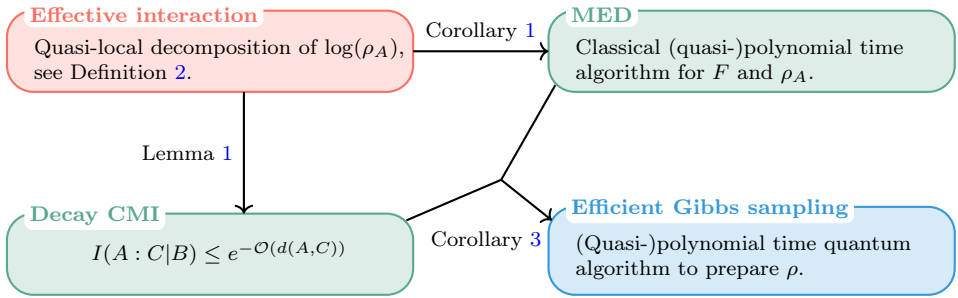

By nature of being a convex relaxation algorithm, the optimization variable in the algorithm does not necessarily correspond to a consistent global state.
A natural question is then to find a \emph{rounding procedure} to obtain a globally consistent state from the locally consistent ones.
This idea is central to the study of approximation ratios in combinatorial optimization \cite{williamson2011design}.
Here, we follow this approach, but the rounding algorithm is a \emph{quantum algorithm} as it should output a large quantum state. More specifically, we use the approximate marginals returned from the MED approximation to construct a sequence of recovery channels that produces a global state close to the Gibbs state. These maps are based on the rotated Petz recovery map, which was introduced as an explicit construction of a map that occurs in a strengthened data-processing inequality \cite{fawzi_quantum_2015, junge2018}.
The resulting map turns out to be \emph{efficiently implementable}, providing an \emph{efficient Gibbs sampler} in the settings described above. The combination of our classical approximation scheme with our quantum rounding procedure thus neatly combines two known approaches to exploit the Markov property of Gibbs states in algorithm design.

It should be noted however that our quantum rounding algorithm can be combined with \textit{any} classical algorithm that computes approximations of the marginals and is not limited to the MED. The obtained guarantee then depends on the approximation quality of the marginals and on a bound on the decay of the CMI.

\paragraph{Outline}
The remainder of this introduction gives a technical overview of the MED recalling the main result from \cite{Poulin2011} and outlines our main result in an informal way.
We illustrate our main result using the special case of exactly vanishing CMI in Section~\ref{sec:exact_case}.
Section~\ref{sec:general_theory} proves our main result, the efficiency of the MED for systems with decaying effective interactions, whose existence in the settings mentioned above is shown in Section~\ref{sec:efficiency}.
We conclude with Section~\ref{sec:rounding}, which introduces our rounding scheme and its implications for Gibbs sampling.

\subsection{The Markov Entropy Decomposition}

\paragraph{Setup} We consider a finite spin lattice system $\Lambda\subset\ZZ^D$, $|\Lambda|=N$  with Hilbert space $\mathcal H=\otimes_{v\in\Lambda}\mathbb C^d$, and a geometrically local Hamiltonian given by terms $h_X$, which have support on $X\subset\Lambda$. The distance function on the lattice $\Lambda$ is denoted by $d(\cdot,\cdot)$, and for sets $A,B\subset\Lambda$ we define $d(A,B)=\min_{x\in A,y\in B}d(x,y)$ and $B_\ell(A)=\{y\in\Lambda:d(y,A)\le\ell\}$.
The Hamiltonian for a subset $V \subset \Lambda$ is given by
\[
H_{V}=\sum_{X\subset V} h_X \, ,
\]
and we define the Gibbs state as
\[
\rho=\exp(-H)/\Tr[\exp(-H)] \, ,
\]
abbreviating the system Hamiltonian $H:=H_\Lambda$.
We assume a finite range of interactions with all non-vanishing terms $h_X$ fulfilling $\diam(X)\le r$ for a constant range $r$.
We absorb the temperature $\beta$ into the definition of the Hamiltonian terms. The \textit{free energy} is given by the following minimization formula
\begin{equation}
    \label{eq:def-free-energy}
\F=\min_{\sigma:\sigma\ge0,\Tr[\sigma]=1}\F(\sigma)=\min_{\sigma:\sigma\ge0,\Tr[\sigma]=1} \Tr[H\sigma] - S(\sigma) \, ,
\end{equation}
where $S(\sigma) = -\Tr[\sigma\log(\sigma)]$ denotes the von Neumann entropy.
Since it involves an optimization over a density matrix which is exponentially large in system size, it cannot be used directly for numerical approaches.

\paragraph{Markov shields} In order to define the Markov entropy decomposition, we first fix a numbering $1,\ldots,N=|\Lambda|$ of the sites of the lattice. In 1D we assume a consecutive ordering, whereas in higher dimensions no such assumption is needed. We choose, for each site $k \in \{1,\ldots,N\}$ a \emph{Markov shield} $S_k\subset [k]=\{1,\ldots,k\}$ which is a subset of the sites that appear before site $k+1$. Even though the MED can be defined for any choice of shields, in this paper we will assume that the shields are given by
\begin{equation}
\label{eq:Skshield}
S_k = B_{\ellshield}(\{k+1\}) \cap \{1,\ldots,k\} \, ,
\end{equation}
where $\ellshield$ is a parameter governing the shield size, see Figure \ref{fig:markovshield2d}. We will further choose $\ellshield$ such that $\ellshield \geq r$ (the locality of the Hamiltonian), which guarantees that
\begin{equation}
\label{eq:hk}
h_k:=H_{[k+1]}-H_{[k]}
\end{equation}
is supported in $S_k \cup \{k+1\}$. For convenience, we will use the notation
\begin{equation*}
S'_k = S_k \cup \{k+1\}
\end{equation*}
and we take the convention $S_0 = \emptyset$.

\begin{figure}[ht]
\centering
    \includegraphics[scale=0.35]{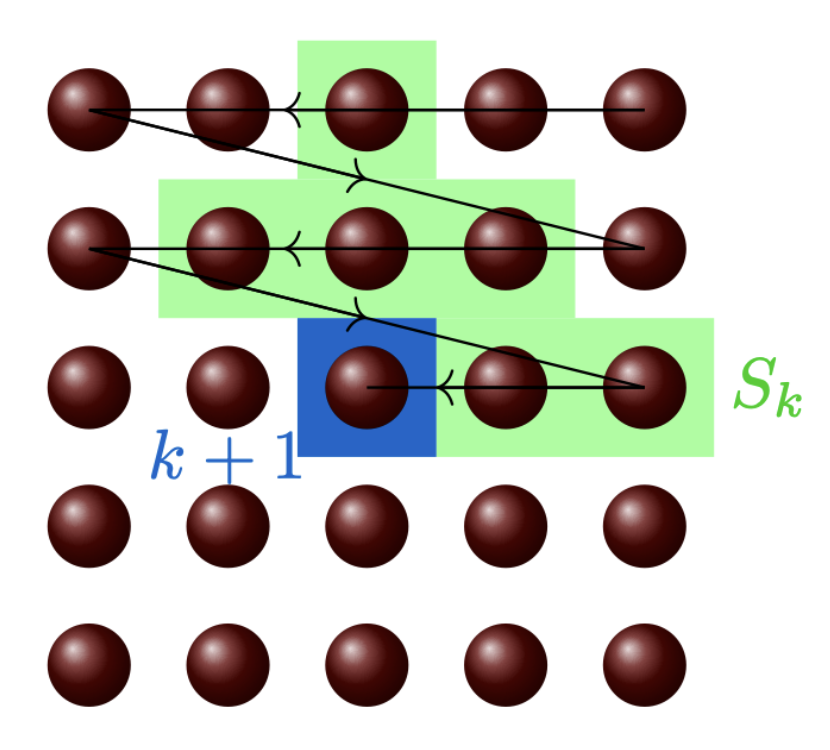} 
    \caption{Illustration of the Markov shield $S_k$ as defined in \eqref{eq:Skshield} with $\ellshield=2$.}
    \label{fig:markovshield2d}
\end{figure}

\paragraph{Markov Entropy Decomposition} Given a set of density matrices $\{\sigma_{S'_k}\}_{1\leq k \leq N-1}$ defined on the regions $S'_k$ for $k \in \{1,\ldots,N-1\}$, we define the Markov entropy decomposition functional as
\begin{equation*}
\FMED\left(\{\sigma_{S_k'}\}\right)=\sum_{k=0}^{N-1} \Tr[h_k\sigma_{S_k'}]-\sum_{k=0}^{N-1} S(k+1|S_k)_{\sigma_{S_k'}},
\end{equation*}
where $h_k$ is as defined in \eqref{eq:hk} and  $S(k+1|S_k)_{\sigma_{S_k'}}$ denotes the conditional entropy, i.e., $S(k+1|S_k)_{\sigma_{S_k'}} = S(\sigma_{S'_k}) - S(\tr_{k+1} \sigma_{S'_k})$. 
Observe that if the marginals $\{\sigma_{S_k'}\}$ are consistent with a global quantum state $\sigma$, i.e., $\sigma_{S_k'}=\tr_{\Lambda\setminus{S_k'}}[\sigma]$, then $\FMED(\{\sigma_{S_k'}\}) \leq \F(\sigma)$ by strong subadditivity, since:
\begin{align}
    \FMED(\{\sigma_{S_k'}\})-\F(\sigma)&=S([N])_\sigma-\sum_{k=0}^{N-1} S(k+1|S_k)_\sigma\label{eq:FMED-F}\\
    &\leq S([N])_\sigma-\sum_{k=0}^{N-1} S(k+1|[k])_\sigma \label{eq:SSA}\\
    &= 0.
\end{align}
We denote subsystem entropies by $S(A)_\sigma=-\Tr[\sigma\log(\Tr_{\overline{A}}[\sigma])]$.
The MED relaxation is then defined as the minimization of the functional $\FMED$ over all \emph{locally consistent} collections of density matrices $\{\sigma_{S'_k}\}_{1\leq k\leq N-1}$, namely:
\begin{align}
\label{eq:defFMED}
\FMED=\min_{\{\sigma_{S_k'}\} \textrm{ s.t. }B\left(\{\sigma_{S_k'}\}\right)=0} \FMED\left(\{\sigma_{S_k'}\}\right),
\end{align}
where we suppressed the implicit constraints $\sigma_{S_k'}\ge0$ and $\Tr\left[\sigma_{S_k'}\right]=1$, and $B$ is a linear map that encodes the local consistency constraints, more precisely
\begin{align}
B:\bigoplus_k \mathcal B\left(\mathcal{H}_{S_k'}\right)&\to\bigoplus_{j,l} \mathcal B\left(\mathcal{H}_{S_j'\cap S_l'}\right)\\
B\left(\{\sigma_{S_k'}\}\right)_{j,l}&=\tr_{S'_j \setminus S'_l}\left[\sigma_{S_j'}\right]-\tr_{S'_l \setminus S_j'}\left[\sigma_{S_l'}\right] \label{eq:expression_B}\,,
\end{align}
i.e., Eq.~\eqref{eq:expression_B} vanishes if and only if for every pair of density matrices the respective reduced density matrices on the intersection of their supports are equal. This constraint trivially holds for a set of marginals of a consistent global state.

It is immediate from the observation \eqref{eq:FMED-F}-\eqref{eq:SSA} above and the definition of the free energy \eqref{eq:def-free-energy} that 
\begin{equation}\label{eq:FMED_bounded_F}
    \FMED\le \F.
\end{equation}

\subsection{Main contributions}

In this paper we prove an a priori guarantee on the convergence of the MED relaxation, and we give a rounding procedure to recover a global quantum state from the approximate marginals of the MED.

\paragraph{Effective interaction and convergence of MED} Our theorem on the convergence of the MED is based on a condition on the \emph{effective interaction}. Given a region $A \subset \Lambda$, the parent Hamiltonian of the marginal of the Gibbs state on $A$ can be defined as $\tilde{H}_A = -\log(\rho_A) - \log(\Tr e^{-H})$ so that $\rho_A = e^{-\tilde{H}_A}/\Tr e^{-\tilde{H}_A}$. (This is sometimes known as the Hamiltonian of mean force or effective Hamiltonian, see \cite[Section V.B]{alhambra2023quantum}.) The additive constant in the definition is for convenience as it cancels the partition function in the definition of $\rho_A$, but not necessary to obtain our results as we make no assumptions on the overall strength of the effective interaction. A natural question is to know how different $\tilde{H}_A$ is from $H_A$, the original Hamiltonian restricted to the region $A$. The exact effective interaction on $A$ is defined as the difference
\begin{equation*}
    \heffexact{A} = \tilde{H}_A - H_A \, .
\end{equation*}
Informally speaking, we say that a Hamiltonian has an exponentially-decaying effective interaction if $\heffexact{A}$ can be decomposed into local terms whose norms decay exponentially with the size of their support, see Definition~\ref{def:heff} for a formal statement. Such an assumption is known to hold for 1D local Hamiltonians at any constant temperature 
as well as for $D$-dimensional local Hamiltonians with commuting terms and marginals 
above a critical temperature, see Section~\ref{sec:efficiency}.\footnote{It is an interesting open question to prove the decay of effective interactions in more general settings. Results on the decay of the CMI (a weaker condition) do exist, see e.g.,~\cite{kato2019quantum} and~\cite{castelnovo2007entanglement,castelnovo2008topological} for specific models at low temperatures.
The paper \cite{Kuwahara2024} proves the decay in higher dimensions at all temperatures, but in a slightly weaker version than needed for our applications.
}
Our first main result is an efficiency guarantee for the MED under the above assumption.
\begin{theorem}[Informal version of Theorem~\ref{thm:errorEff}]
\label{thm:mainthmMEDintro}
Let $\rho$ be the thermal state, and $\F$ the free energy of a Hamiltonian with exponentially-decaying effective interaction defined on a lattice $\Lambda$ with $N$ sites. Let $\eps>0$.
Then for some choice of $\ell_S=\mathcal{O}(\log(N/\eps))$, the Markov Entropy Decomposition \eqref{eq:defFMED} outputs approximations $\FMED$, $\sigma_X$ for constant sized sets $X \subset \Lambda$ such that
\begin{align*}
|\FMED-\F|&\le\eps \, ,\\
\|\rho_X-\sigma_X\|_1&\le\eps\,.
\end{align*}
\end{theorem}
The theorem refers to the exact minima of the MED.
However, the MED is a convex optimization problem of size $\exp(\mathcal O(\ell_S^D))$ such that we can turn the above into a classical approximation algorithm  of polynomial (1D) or quasi-polynomial (higher D) runtime in $N$ and $1/\eps$, see Corollary~\ref{cor:free_energy_algorithm}.

\paragraph{Rounding} Our second main result concerns the rounding scheme, giving a quantum algorithm that reconstructs a global state from approximate local marginals using rotated Petz recovery maps. In fact our result is not restricted to the setting of Gibbs states or the MED and applies generally to any density matrix with decaying CMI.
\begin{theorem}[Informal version of Theorem~\ref{thm:rounding}]
\label{thm:mainthmRoundingintro}
Let $\rho$ be a state on a lattice $\Lambda$ with $N$ sites such that
\[
I(A:C|B)\le \exp(-\Omega(d(A,C)))\,\qquad\forall \textrm{ disjoint }A,B,C\subset \Lambda.
\]
For a given $\eps>0$, there is an input error $1/\eps_\sigma=\poly(1/\eps)$ and $\ellshield=\mathcal{O}(\log(N/\eps))$ such that given marginal approximations $\sigma_X$ with $\|\sigma_X-\rho_X\|\le\eps_\sigma$ on all regions $S_k'$, there are channels $(\phi_k)_{2\leq k\leq N}$ that reconstruct the global state
\[
\|\widetilde\sigma-\rho\|_1=\left\|\phi_N \circ \dots \circ \phi_2(\sigma_{\{1\}})-\rho\right\|_1\le\eps,
\]
where each channel $\phi_k$ only acts on sites $S'_k$.
\end{theorem}
Again, this theorem, which only gives a channel construction, can be lifted to a result on a quantum circuit preparing the Gibbs state.
More specifically, classically numerically computing a representation of the channels and decomposing them into quantum gates gives (quasi-)polynomial classical and quantum runtimes for the preparation of Gibbs states under the decay of the effective interaction assumption, see Corollary~\ref{cor:circuitDecomposed}.

\begin{remark}
We would like to comment on the omitted temperature dependence in our results.
Our proofs in fact have no direct temperature dependence but a temperature dependence does enter through the respective results on the existence of effective interactions.
Their decay rate itself is a decreasing function of temperature.
Our results and in particular the big-$\mathcal{O}$ notation should be understood as only valid for \emph{constant temperatures}.
For the high temperature results, the statements break down below a critical temperature and the runtime diverges as the temperature approaches this point.
In 1D, no such critical temperature exists and the results hold for any fixed temperature, but a scaling to obtain results on ground states is still impossible with the correlation length growing doubly exponentially in $\beta$.
\end{remark}

\paragraph{Related work}
In our work, we demonstrate various novel connections between physical concepts of decaying correlations, information theoretic insights into many-body states and the algorithmic complexity of problems related to thermal states.
The resulting asymptotic runtimes, however, are known to be achievable through a number of previously proposed algorithms, most of which are quite different in nature.
We discuss some of the most closely related formulations in the following.

Regarding the MED itself, rigorous convergence guarantees were only known in one setting for a closely related algorithm \cite{Bravyi2022}, which addresses a setup of dense 2-local Hamiltonians. This work runs the MED but optimizes over the choice of shields. The resulting algorithm is polynomial-time in system size but not in the inverse error.

For the free energy and marginal approximations many prior works have addressed the computational complexity of these problems; however, the approaches are usually quite different for the respective settings.
In one dimension, a series of works proposed efficient algorithms of polynomial and even subpolynomial runtimes \cite{kuwahara2018,Fawzi2023,Scalet2024,Kuwahara2021,achutha2024}.
These works are based on technical tools specific to the 1D setting including the quantum belief propagation \cite{Hastings_2007,Kim2012} and bounds on imaginary time evolution \cite{Araki1969,Perez2023}. Note that while the structure of the MED as an algorithm is entirely different to all these works, the convergence proof is based on \cite{Kuwahara2024}, which makes heavy use of the quantum belief propagation. For higher dimensions, without commutativity restrictions, a polynomial-time algorithm for the free energy is given in \cite{Mann_2021} (see also \cite{Harrow.2020,Molnar2015}), which is based on a direct evaluation of the cluster expansion.
Again, while the formulation of the algorithm is entirely different, the method using polymer models also inspired the proofs of the existence of effective interactions in \cite{bluhm2024} and thus indirectly leads to our convergence guarantees in the high-temperature setting.
General algorithms treating arbitrary systems exist \cite{Fawzi2024,Kliesch_2014}, but come with more restricted convergence guarantees.

Since our results also have implications for the existence of efficient Gibbs samplers, we also review some results for this problem.
In the 1D setting an efficient Gibbs sampler based on recursively merging Gibbs states of smaller subchains has been proposed in \cite{bilgin2010}.
Apart from that, but only restricted to the translation-invariant commuting setting, approaches based on simulating Lindbladian evolutions  which converge to the Gibbs state can be found in \cite{Bardet.2023,Bardet.2024,Kochanowski.2024}.
Lindbladian based approaches are also successful in the high-temperature regime, where efficiency guarantees for commuting Hamiltonians have been shown in \cite{capel2024,Kochanowski.2024,Kastoryano.2016,Capel.2020}.
Beyond the commuting setting, where the Davies generator does no longer lead to a local Lindbladian, a recent proposal introduced an alternative Lindbladian \cite{Chen.2023a,Chen.2023} which converges to the Gibbs state and can be efficiently implemented.
This approach has been equipped with mixing time bounds and thereby efficiency guarantees in \cite{Rouze.2024b,Rouze.2024c}.
Moving away from the Lindbladian approach other algorithms with efficiency guarantees at high temperature include sampling approaches based on the separability of these states \cite{Bakshi.2024} and adiabatic evolutions \cite{Ge_2016}.
Our approach in Theorem \ref{thm:rounding} using recovery maps has many similarities to  \cite{brandao2019,chen2025quantumgibbsstateslocally}, but there are some subtle differences in the necessary conditions that we discuss in Section \ref{sec:comparisonBrandao}.
Other approaches to noncommutative Gibbs sampling can be found, where polynomial time guarantees depend on currently unknown mixing times \cite{Temme.2011,Jiang.2024,Gilyen.2024,Zhang.2023}, technical assumptions \cite{Fang.2024,Motta.2019}, or the partition function being of the same order as the Hilbert space dimension \cite{Poulin.2009,Chowdhury.2017,Gilyen.2019}.

On a more conceptual level, we note that the idea of using physical or information theoretic results for the design of efficient algorithms has been demonstrated for various quantum many-body problems. This includes the use of area laws \cite{Hastings_2007_arealaw, arad2013,Verstraete2006}, correlation functions \cite{kuwahara2018,Capel.2020}, or conditional mutual informations \cite{Gondolf2024}. An operator version of the latter, which is closely related to our definition of the effective interaction has also found applications to Gibbs sampling \cite{capel2024}.

\section{Warmup: the exact case}
\label{sec:exact_case}

For illustration, we present a simplified version of our main result in this section. We argue that for exact quantum Markov chains, the MED approximation introduced before exactly captures the true free energy, i.e., a vanishing conditional mutual information implies an efficient algorithm for the free energy.
The proof strategy deviates from the approximate case and exploits the commutativity of Hamiltonian terms that follows only in the exact case for all tree graphs \cite{Poulin2011}.
The result in \cite{Poulin2011} can be seen as a restricted quantum analog to the Hammersley-Clifford theorem, which states the equivalence of Markov networks and Gibbs states on arbitrary graphs.
Note that in one dimension, for classical systems, the free energy can be calculated using the transfer matrix method, a standard result in statistical mechanics.
The restriction to 2-locality in the following theorem is without loss of generality since higher locality interactions can always be blocked to range two in one dimension.

\begin{theorem}
\label{thm:free_energy_algorithm_exact}In 1D, given a 2-local Hamiltonian, such that its thermal state $\rho$ fulfills
\begin{equation*}
    I(A:C|B)=0
\end{equation*}
for any adjacent intervals $ABC$ with nonempty $B$, the MED exactly coincides with the free energy
\begin{equation*}
    \FMED=\F.
\end{equation*}
\end{theorem}

\begin{figure}[h]
    \centering
    \includegraphics[width=0.2\linewidth]{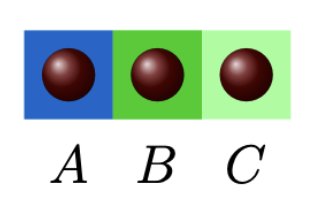}
    \caption{System $ABC$ in a chain.}
    \label{fig:chain4}
\end{figure}

\begin{proof}
In the following, we are making use of the equivalence of exact quantum Markov chains and Gibbs states of commuting Hamiltonian proven in \cite{Poulin2011}.
    Let us first assume for simplicity that we can split the chain in $3$ systems, $\Lambda = ABC$ as in Figure \ref{fig:chain4}, with all systems of the same size, and argue later for an $N$-partite system. We can thus write $H=H_\Lambda = h_{AB} + h_{BC}$, where $[h_{AB},h_{BC}]=0$. Let us recall that the free energy can be computed by
\begin{equation}\label{eq:free_energy}
\F=\min_{\sigma_\Lambda \, : \, \sigma_\Lambda \ge0,\Tr[\sigma_\Lambda]=1} \big( \, \Tr[ h_\Lambda \sigma_\Lambda]-S(\Lambda)_\sigma \, \big) \, ,
\end{equation}
where $S(\Lambda)_\sigma = - \Tr[\sigma_\Lambda \log \sigma_\Lambda]$ is the von Neumann entropy of $\sigma_\Lambda$. As mentioned in Eq. \eqref{eq:SSA}, we can relax the second part of the optimization problem in Eq. \eqref{eq:free_energy} to von Neumann entropies in smaller subsystems using strong subadditivity, and taking into account the decomposition of the Hamiltonian we can consider the following convex program:
    \begin{equation*}
\FMED=\min_{\{ \sigma_{AB} \, , \, \sigma_{BC} \} \text{ l.c.m.}} \big( \, \Tr[ h_{AB} \sigma_{AB}+ h_{BC} \sigma_{BC}]-( S(AB)_\sigma+S(BC)_\sigma - S(B)_\sigma ) \, \big) \, ,
\end{equation*}
    where the minimization is done over pairs of states $\sigma_{AB}  $ and $\sigma_{BC}$ in $AB$ and $BC$, respectively, such that $\Tr_A[\sigma_{AB}]=\Tr_C[\sigma_{BC}]$, i.e., with locally consistent marginals. 

    Next, note that for any two pairs of density matrices $\{\sigma_{AB},\sigma_{BC}\}$ and  $\{\eta_{AB},\eta_{BC}\}$, we have
        \begin{align}\label{eq:ansatz_relative_entropies}
            D\left(\sigma_{AB} \| \eta_{AB} \right) + D\left(\sigma_{BC} \| \eta_{BC}  \right) - D\left(\sigma_{B} \| \eta_{B}  \right) \ge 0 \, 
        \end{align}
 as a consequence of non-negativity of the relative entropy between two density matrices and data-processing inequality. This will be required later in the proof. 

 Consider $\rho_{ABC}$ the Gibbs state of $H$:
 \begin{align*}
            \rho_{ABC} &= \frac{1}{\mathcal Z}e^{-\left( h_{AB}+h_{BC} \right)}  
                = \frac{1}{\mathcal Z}e^{-h_{AB}}e^{-h_{BC}}
        \end{align*}
        where $\mathcal Z$ is a normalization factor. The marginal of $\rho_{ABC}$ on $AB$ is 
        \begin{align}
            \rho_{AB} &=\frac{1}{\mathcal Z} \Tr_C\big[e^{-h_{AB}}e^{-h_{BC}} \big] = \frac{1}{\mathcal Z} e^{-h_{AB}} \Tr_C \big[ e^{-h_{BC}} \big] \, ,\label{eq:partialtraceD}
        \end{align}
        and analogously 
        \begin{align}
            \rho_{BC}  = \frac{1}{\mathcal Z} \Tr_A \big[ e^{-h_{AB}} \big] e^{-h_{BC}} \, ,\label{eq:partialtraceA}
        \end{align}
        and
        \begin{align}
            \rho_{B}  = \frac{1}{\mathcal Z} \Tr_A \big[ e^{-h_{AB}} \big]  \Tr_C \big[ e^{-h_{BC}} \big]\, .\label{eq:partialtraceAD}
        \end{align}
        In particular, note that all the components of each of the RHS of Eq. \eqref{eq:partialtraceD}, Eq. \eqref{eq:partialtraceA} and Eq. \eqref{eq:partialtraceAD} commute.  The marginals of the Gibbs state are clearly locally compatible. Consider a
        pair of locally consistent marginals $\{\sigma_{AB},\sigma_{BC}\}$; then, we can compute the following relative entropies between them and the marginals of the Gibbs state $\rho_{ABC}$:
                \begin{align*}
        D\left(\sigma_{AB}\|\rho_{AB}\right) &= \log\mathcal Z - S(AB)_\sigma - \Tr\left[\sigma_{AB}\left( -h_{AB}+\log\Tr_C \big[ e^{-h_{BC}} \big] \right) \right], \\
            D\left(\sigma_{BC}\|\rho_{BC}\right) &= \log\mathcal Z - S(BC)_\sigma - \Tr\left[\sigma_{BC}\left( -h_{BC}+\log \Tr_A \big[ e^{-h_{AB}} \big] \right) \right], \\
            D\left(\sigma_{B}\|\rho_{B}\right) &= \log\mathcal Z - S(B)_\sigma - \Tr\left[\sigma_{B}\left(\log \Tr_A \big[ e^{-h_{AB}} \big]+\log \Tr_C \big[ e^{-h_{BC}} \big] \right) \right].
        \end{align*}
Replacing them into Eq. \eqref{eq:ansatz_relative_entropies} and using the fact that $\{\sigma_{AB}, \sigma_{BC}\}$ are locally consistent, we obtain
      \begin{align*}
            \log\mathcal Z - S(AB)_\sigma - S(BC)_\sigma + S(B)_\sigma + \Tr\left[\sigma_{AB} \, h_{AB}+ \sigma_{BC} \, h_{BC}\right] \ge 0.
        \end{align*}
        Using the fact that $\F = -\log \mathcal{Z}$, we get
        \begin{align*}
            \F \le \Tr\left[\sigma_{AB} \, h_{AB}+ \sigma_{BC} \, h_{BC}\right] - (S(AB)_\sigma + S(BC)_\sigma - S(B)_\sigma ),
        \end{align*}
        with equality if, only if, $  \{\sigma_{AB},\sigma_{BC}\} = \{\rho_{AB},\rho_{BC}\}$ (as this is the unique case in which equality in Eq. \eqref{eq:ansatz_relative_entropies} holds). Therefore, 
        \begin{equation*}
            \F= \FMED \, ,
        \end{equation*}
        and it is attained for $  \{\sigma_{AB},\sigma_{BC}\} = \{\rho_{AB},\rho_{BC}\}$.

\begin{figure}
    \centering
    \includegraphics[width=0.35\linewidth]{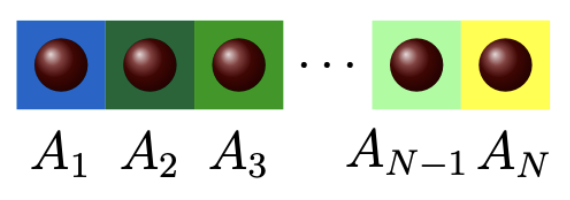}
    \caption{$N$-partite system $A_1 \ldots A_N$ in a chain.}
    \label{fig:chainN}
\end{figure}
        To extend this result to $N$-partite systems $\Lambda = A_1 \ldots A_N$ as in Figure \ref{fig:chainN}, the procedure is exactly analogous as before, where now  
        \begin{equation*}
        \FMED=\min_{\{ \sigma_{A_iA_{i+1}}  \} \text{ l.c.m.}} \left( \, \underset{i=1}{\overset{N-1}{\sum}} \Tr[ h_{A_iA_{i+1}} \sigma_{A_iA_{i+1}}]-\bigg( \underset{i=1}{\overset{N-1}{\sum}}S(A_{i}A_{i+1})_\sigma- \underset{i=2}{\overset{N-1}{\sum}} S(A_{i})_\sigma \bigg) \, \right) \, ,
        \end{equation*}
        and the Ansatz required to prove that $\F=\FMED$ (as a consequence of its equality conditions)  is 
\begin{equation*}
    \underset{i=1}{\overset{N-1}{\sum}} D(\sigma_{A_iA_{i+1}} \| \eta_{A_iA_{i+1}}) - \underset{i=2}{\overset{N-1}{\sum}} D(\sigma_{A_i} \| \eta_{A_i}) \geq 0 \, ,
\end{equation*}    
     evaluated in a set of locally consistent marginals  $\{ \sigma_{A_iA_{i+1}}  \}$, i.e., $\Tr_{A_{i}}[\sigma_{A_iA_{i+1}}] = \Tr_{A_{i+2}}[\sigma_{A_{i+1}A_{i+2}}]$ for $i=1, \ldots, N-2$, and $\eta_X$ the marginal $\rho_X$ of the Gibbs state for every $X$.  
\end{proof}

\begin{remark}
    As a consequence of Theorem \ref{thm:free_energy_algorithm_exact}, there exists an algorithm taking as input a description of the terms of a local, commuting Hamiltonian $h_X$ in 1D and outputting the free energy $\F$ 
and the marginals $\rho_A$ on subsets $A$ (of constant size) of the thermal state up to an error $\eps$
in time
\begin{equation*}
    \poly(N,\log(1/\eps)) \, ,
\end{equation*}
 This is due to the fact that the convex optimization stores $N$ density matrices of constant size, and thus, the above runtime can be derived by standard theorems for convergence of convex optimization (e.g.  \cite[Corollary 2]{Bravyi2022}). 
\end{remark}

\section{Convergence guarantees of the MED}
\label{sec:general_theory}

In this section we prove a convergence guarantee on the MED relaxation under the exponentially-decaying effective interaction assumption.
The precise statement of the latter assumption is given in Definition~\ref{def:heff}.
The main argument in this section can be found in Lemma~\ref{lem:MEDerror}, which bounds the approximation error of the MED in terms of its gradient and Lemma~\ref{lem:gradient}, which bounds the projection of this gradient into the constraint subspace of the optimization.
The following Theorem~\ref{thm:free_energy_algorithm_exact} and Corollary~\ref{cor:free_energy_algorithm} translate the result into concrete error and runtime bounds respectively.

\subsection{The exponentially-decaying effective interaction assumption}

For a finite lattice $\Lambda \subset \mathbb{Z}^D$, a local Hamiltonian $H \equiv H_\Lambda$ supported on $\Lambda$ and $A \subset \Lambda$, the main object of interest in the derivation of these conditions is the so-called effective interaction.
\begin{definition}\label{def:heffexact}

    Given a local Hamiltonian $H \equiv H_\Lambda$ in a finite lattice $\Lambda$, the effective interaction associated to it is given by 
    \begin{equation}\label{eq:heffexact}
    \heffexact{A}=-H_A-\log(\rho_A)-\log(\Tr\left[e^{-H}\right]) \, ,
\end{equation}
for any $A \subset \Lambda$. 
\end{definition}

We will make an assumption on the effective interaction closely related to properties like locality of temperature, decay of correlations and the decay of the conditional mutual information. 

\begin{definition}\label{def:heff}
We say that a local Hamiltonian on a finite lattice $\Lambda$ has an $(\eps,\ell)$-effective interaction if there exists a set of effective interactions with local decompositions
\[
\hefftot{A}=\sum_{X\subset \Lambda}\heff{A}{X} \quad\in\mathcal{B}(\mathcal H_A)
\]
for each $A \subset \Lambda$, where each $\heff{A}{X}\in\mathcal B(\mathcal H_{X\cap A})$ is such that
\begin{enumerate}[(i)]
\item $\left\|\hefftot{A} -\heffexact{A}\right\|\le\eps \, ;$
\item $\heff{A}{X}=0$ if $X \subset \Lambda$ with $\diam(X)\ge\ell $  $\, ;$
\item $\heff{A}{X}=0$ if $X \subset \Lambda$ and $\exists i\in X$ such that $i\notin A$ or $d(i,\bar A)\ge\ell \, ;$ 
\item $\heff{A}{X}=\heff{AB}{X}$ if $d(X,B)\ge\ell \,$.
\end{enumerate}
In the 1D case, we require the $\hefftot{A}$ only for $A$ connected.
\end{definition}
\begin{definition}\label{def:heffexp}
We say that a local Hamiltonian has an exponentially-decaying effective interaction if there are constants $C,\xi\ge0$ such that for each $\ell$, it has a $(|\Lambda|C\exp(-\ell/\xi
),\ell)$-effective interaction. In 1D we assume a stronger $(C\exp(-\ell/\xi
),\ell)$-effective interaction.
\end{definition}

As we will see in Section~\ref{sec:efficiency}, the existence of an exponentially-decaying effective interaction can be rigorously proven for one-dimensional systems.
Furthermore, the existence of an effective interaction in arbitrary dimension at high temperature was thought for some time to be proven in \cite{Kuwahara2020}, but there is a gap in the proof. 
However, under additional assumptions, namely a commuting Hamiltonian and commuting marginals of the Hamiltonian terms, the result can be recovered from \cite{bluhm2024} which we recall in Section~\ref{sec:efficiency}. We still expect that this result should hold in more generality, for instance, in the Toric code model in two dimensions at low but nonzero temperature~\cite{castelnovo2007entanglement}, but leave it as an open question.

We defer the main existence proofs to Section~\ref{sec:efficiency}, but present here the simple example of 1D commuting 2-local Hamiltonians, which shows that the general theory in this section recovers the result in Section~\ref{sec:exact_case}:
\begin{example}\label{example:commuting1D}

Let $\Lambda=ABC=[1,N]$ with subsequent intervals $ABC$, $B=[a,b]$ and a Hamiltonian defined by the terms $h_{\{i,i+1\}}$. We find
\[
\log(\rho_B)+\log(\Tr[e^{-H}])=\log(\Tr_A[e^{-H_A-h_{\{a-1,a\}}}])-H_B+\log(\Tr_C[e^{-H_C-h_{\{b,b+1\}}}])\, ,
\]
so a (0,2)-effective interaction is given by defining the only nonzero terms
\begin{align*}
    \heff{B}{\{a\}}&=-\log(\Tr_A[e^{-H_A-h_{\{a-1,a\}}}])\,,\\
    \heff{B}{\{b\}}&=-\log(\Tr_C[e^{-H_C-h_{\{b,b+1\}}}])\, .
\end{align*}
\end{example}

We also prove the following Lemma, which shows that decay of the effective interaction is a stronger condition than the decay of CMI (see \cite{Kuwahara2024,capel2024} for closely related statements). It is not part of the convergence proof of the MED, however, a good warmup exercise since the combinatorial proof idea will be similiar to that of Lemma~\ref{lem:gradient}.
The decay of CMI will be the sufficient condition for the rounding scheme that we present in Section~\ref{sec:rounding}.

We want to point out that despite several recent results on the decay of conditional mutual information, the following result remains a promising route to novel insights in this area by means of extending the results on the decay of the effective interaction. In comparison to other approaches \cite{Kuwahara2024,chen2025quantumgibbsstateslocally} it should be noted that our result does not suffer from exponential overheads in either $|A|$ or $|C|$, nor requires $ABC$ to coincide with $\Lambda$.

\begin{lemma}\label{lem:CMI_decay}
If a local Hamiltonian $H$ on $\Lambda$ has an exponentially-decaying effective interaction according to Definition \ref{def:heff}, then the Gibbs state $\rho$ of $H$ satisfies exponential decay of the conditional mutual information: there exist constants $D, \eta$ (that only depend on the locality of the Hamiltonian and on the parameters of the effective interaction) such that for any disjoint sets $A, B, C \subset \Lambda$, we have
\[
I(A:C|B)_{\rho} \le |\Lambda|D \exp(-d(A,C)/\eta)\, 
\]
In 1D we assume that $A$, $B$, $C$ are adjacent intervals with $B$ separating $A$ from $C$ and obtain a stronger decay result
\[
I(A:C|B)_{\rho} \le D \exp(-d(A,C)/\eta)\,.
\]
\end{lemma}
\begin{proof}
Without loss of generality we assume $d(A,C)>r$, i.e., there are no
 interaction terms with support intersecting $A$ and $C$\footnote{The case $d(A,C)\le r$ can be handled by simply choosing a sufficiently large constant $D$, since it can depend on $r$. }.
We consider an $(\eps,\ell)$-effective interaction for $\ell=\lfloor d(A,C)/3\rfloor$.
Let us first note that
\begin{align}
I(A:C|B)_{\rho} &=\Tr[\rho_{ABC}(\log\rho_{ABC}+\log\rho_B-\log\rho_{AB}-\log\rho_{BC})]\\
&\le\|\log\rho_{ABC}+\log\rho_B-\log\rho_{AB}-\log\rho_{BC}\|\\
&\le4\eps+\|H_{ABC}-H_{BC}-H_{AB}+H_B\|+\|\hefftot{ABC}-\hefftot{AB}-\hefftot{BC}+\hefftot{B}\|\label{eq:boundCMIFromOp}
\end{align}
We first prove $\|H_{ABC}-H_{BC}-H_{AB}+H_B\|=0$ by showing that each term $h_X$ is included in equally many of $H_{ABC}, H_B$ as $H_{AB}, H_{BC}$.
If $X\not\subset ABC$ it occurs in neither term.
Otherwise if $X\cap A\neq\emptyset$, $X\cap C=\emptyset$ and $h_X$ appears in $H_{ABC}$ and $H_{AB}$, analogously if $X\cap C\neq\emptyset$.
The remaining case is that $X\subset B$, which means $h_X$ appears in all four terms.

Secondly, we show $\|\hefftot{ABC}-\hefftot{AB}-\hefftot{BC}+\hefftot{B}\|=0$ using a similar argument, comparing and canceling for any $X\subset\Lambda$ all nonzero terms $\heff{ABC,AB,BC,B}{X}$. 
If $d(X,A)<\ell$, by (ii) $d(X,C)\ge d(A,C)-d(X,A)-\diam(X)\ge\ell$. Therefore, by (iv) $\heff{ABC}{X}=\heff{AB}{X}$ and $\heff{BC}{X}=\heff{B}{X}$ (while not all four terms need to be equal). If $d(X,A)\ge\ell$, by (iv) $\heff{ABC}{X}=\heff{BC}{X}$ and $\heff{AB}{X}=\heff{B}{X}$.

So by Eq.~\eqref{eq:boundCMIFromOp}, we have
\[
I(A:C|B)\le4\eps
\]
and since we can choose $\eps\le |\Lambda|C\exp(-\ell/\xi)\le |\Lambda|Ce^{1/\xi}\exp(-d(A,C)/3\xi)$ the Lemma follows.

The requirement in 1D reflects the additional assumption in the 1D case of connected regions and the stronger decay follows directly from the stronger decay assumption, i.e., $\eps\le C\exp(-\ell/\xi)\le Ce^{1/\xi}\exp(-d(A,C)/3\xi)$.
\end{proof}

\subsection{Main lemmas}

\newcommand{\<}{\left\langle}
\renewcommand{\>}{\right\rangle}

The strategy to prove an a priori bound on the MED relaxation is as follows. Since the functional $\FMED$ is convex and differentiable, we know that for any $\{\sigma_{S'_k}\}$ and $\{\rho_{S'_k}\}$ we have
\[
\begin{aligned}
\FMED(\{\sigma_{S'_k}\}) &\geq \FMED(\{\rho_{S'_k}\}) + \<\nabla\FMED(\{\rho_{S'_k}\}), \{ \sigma_{S'_k} - \rho_{S'_k} \}\>\\
&= \<\nabla\FMED(\{\rho_{S'_k}\}), \{\sigma_{S'_k}\}\>-1
\end{aligned}
\]
where the second line is a simple calculation, and follows e.g., from homogeneity of $\FMED$. We choose $\{\rho_{S'_k}\}$ to be the marginals of the true Gibbs state. Since $\{\sigma_{S'_k}\}$ are locally consistent, i.e., $\{\sigma_{S'_k}\} \in \ker(B)$ (where $B$ is the linear map defined in \eqref{eq:expression_B}), we then get that for any choice of $Z = \{Z_{S'_k}\} \in \ker(B)^{\perp}$
\begin{equation}
\label{eq:FMED-Fbound}
\FMED(\{\sigma_{S'_k}\}) - \F \geq \log \Tr e^{-H} + \<\nabla \FMED(\{\rho_{S'_k}\}) + Z, \{\sigma_{S'_k}\} \>-1
\end{equation}
where $\F = -\log \Tr e^{-H}$ is the free energy. The key lemma (Lemma \ref{lem:gradient}) is to show that under the exponentially decaying effective interaction assumption, one can choose $Z \in \ker(B)^{\perp}$ that makes the right-hand side of \eqref{eq:FMED-Fbound} small.

We summarize the above strategy in Lemma \ref{lem:MEDerror} below. We will make use of the following expression for the gradient of $\FMED$:

\begin{equation}
\begin{aligned}
   \left(\nabla \FMED(\{\sigma_{S_k'}\})\right)_{S_0'}&=h_{i}+\log(\sigma_{S_i'})+\id_{1}\\
   \left(\nabla \FMED(\{\sigma_{S_k'}\})\right)_{S_i'}&=h_{i}+\log(\sigma_{S_i'})-\log(\sigma_{S_i})\otimes \id_{i+1} \, . \label{eq:gradient-MED}
   \end{aligned}
\end{equation}
where the index of $\nabla\FMED$ refers to the partial derivative with respect to $\sigma_{S_i'}$.

\begin{lemma}\label{lem:MEDerror}
The following inequalities hold:
\begin{equation}\label{eq:errFDelta}
\FMED\le \F\le \FMED+\delta \, ,
\end{equation}
and
\begin{equation}\label{eq:errDivDelta}
\sum_{k=0}^{N-1} D(\sigma_{S_k'}\| \rho_{S_k'})-D(\sigma_{S_k} \| \rho_{S_k})\le\delta \, ,
\end{equation}
where
\begin{equation*}
\delta =-\log(\Tr\left[e^{-H}\right])-\left\langle\nabla\FMED(\{\rho_{S_k'}\}),\{\sigma_{S_k'}\}\right\rangle+1
\end{equation*}
with the marginals of the thermal state $\rho_{S_k'}$ and the minimizers of the MED $\sigma_{S_k'}$ respectively.
\end{lemma}
\begin{proof}
First, observe that the first inequality in Eq. \eqref{eq:errFDelta} was proven in Eq. \eqref{eq:FMED_bounded_F}. 
For the second one, let us rewrite the objective function in the following way:
\begin{equation*}
    \begin{aligned}
        \FMED(\{\sigma_{S_k'}\})&=\sum_{k=0}^{N-1} \text{Tr}(\sigma_{S_k'} h_k) - \sum_{k=0}^{N-1} S(k+1|S_k)_{\sigma} \\
      & = \underbrace{\left(\sum_{k=0}^{N-1} \text{Tr}[\sigma_{S_k} \log \rho_{S_k}] + S(\sigma_{S_k})\right) -\left(\sum_{k=0}^{N-1} \text{Tr}[\sigma_{S_k'} \log \rho_{S_k'}] + S(\sigma_{S_k'})\right)}_{\text{relative entropy}}\\
      &\phantom{==}- \underbrace{\delta}_{\text{error}}  \underbrace{-\log \text{Tr}[e^{-H}]}_{\text{free energy}}\\
      &\ge \F-\delta,
    \end{aligned}
\end{equation*}
where the error term can be written as 
\begin{equation*}
    \delta =-\log(\Tr\left[e^{-H}\right]) -\sum_{k=0}^{N-1} \text{Tr}(\sigma_{S_k'} (h_{k} + \log \rho_{S_k'} - \log \rho_{S_k})).
\end{equation*}
The inequality follows because the relative entropy term can be rewritten as 
\begin{equation}\label{eq:errFDivDelta}
\begin{aligned}
    \left(\sum_{k=0}^{N-1} \text{Tr}[\sigma_{S_k} \log \rho_{S_k}] + S(\sigma_{S_k})\right) -\left(\sum_{k=0}^{N-1} \text{Tr}[\sigma_{S_k'} \log \rho_{S_k'}] + S(\sigma_{S_k'})\right) \\
    = \left(\sum_{k=0}^{N-1} D(\sigma_{S_k'}\| \rho_{S_k'})-D(\sigma_{S_k} \| \rho_{S_k}) \right),
\end{aligned}
\end{equation}
justifying its name. By the data processing inequality, the relative entropy term must be nonnegative. This shows the upper bound in~\eqref{eq:errFDelta} on the free energy.
To prove~\eqref{eq:errDivDelta} we rearrange Eq.~\eqref{eq:errFDivDelta} and use Eq.~\eqref{eq:FMED_bounded_F}
\begin{align*}
\sum_{k=1}^{N-1} D(\sigma_{S_k'}\|\rho_{S_k'})-D(\sigma_{S_k}\|\rho_{S_k})= \FMED(\{\sigma_{S_k'}\})-\F+\delta\le\delta \, .
\end{align*}
\end{proof}

Since the minimizer of the MED satisfies local consistency, i.e., $B\left(\{\sigma_{S'_k}\}\right) = 0$ the inequality of the previous lemma still holds with
\[
\delta = -\log\left(\Tr[e^{-H}]\right) - \< \nabla \FMED(\{\rho_{S'_k}\}) + Z, \{\sigma_{S'_k}\}\>+1
\]
where $Z = \{Z_{S'_k}\} \in \ker(B)^{\perp} = \im(B^*)$.
Note that for two indices $j,k\in\{0,\ldots,N-1\}$, any vector of the form
\begin{equation}\label{eq:kerConstruction}
Z_{S_j'}=X\otimes\id\qquad Z_{S_k'}=-X\otimes\id\qquad \textrm{ for }X\textrm{ supported on }S_j'\cap S_k'
\end{equation}
with all other elements equal to zero fulfills $Z\in\ker(B)^\perp$. The following key lemma shows that the existence of an effective interaction allows us to construct such a vector $Z$ that approximately cancels the gradient $\nabla \FMED(\{\rho_{S'_i}\})$.

\begin{lemma}\label{lem:gradient}
Given an $(\eps,\ell)$-effective interaction, there exists $Z\in\ker(B)^\perp$, such that
for $\ellshield=5\ell$ 
\begin{align}
\label{eq:propertyZ k=0}
\left\|(\nabla \FMED(\{\rho_{S_k'}\}))_{S'_0}+Z_{S_0'} + \log(\Tr\left[e^{-H}\right])-\id\right\|\le2\eps \, ,
\end{align}
and for $i = 1, \ldots, N-1$
\begin{align}
\label{eq:propertyZ k > 0}
\left\|(\nabla \FMED(\{\rho_{S_k'}\}))_{S'_i}+Z_{S_i'}\right\|\le2\eps \, ,
\end{align}
where $\rho_{S_k'}$ are the marginals of the thermal state.
\end{lemma}
\begin{proof}
The proof will appeal to properties (i)-(iv) of the effective interaction $\hefftot{A}$ in Definition~\ref{def:heff} for $A$ being one of $[i+1]$, $[i]$, $S_i'$, or $S_i$.
For ease of notation, we will also write $\nabla \FMED$ for $\nabla \FMED(\{\rho_{S'_k}\})$ in the remainder of the proof. 

We first show that each component of the gradient can be approximated by a difference of effective interactions. 
We will use the equality
\begin{align}\label{eq:difference-Heff}
\hefftot{[i+1]}-\hefftot{[i]} = \hefftot{S'_{i}}-\hefftot{S_i}\, .
\end{align}
This can be verified by comparing the two sides term by term, i.e., showing $\heff{[i+1]}{X}-\heff{[i]}{X} = \heff{S'_{i}}{X}-\heff{S_i}{X}$.
If $d(X,\{i+1\})<\ell$, then $d(X,[i+1]\setminus S_i')\ge\ell$ by property (ii) and the triangle inequality. Then choosing $B=[i+1]\setminus S_i'=[i]\setminus S_i$, property (iv) implies $\heff{[i+1]}{X} = \heff{S'_{i}}{X}$ and $\heff{[i]}{X} = \heff{S_{i}}{X}$ respectively.
If $d(X,\{i+1\})\ge\ell$, choosing $B=\{i+1\}$ in (iv) ensures $\heff{[i+1]}{X}=\heff{[i]}{X}$ and $\heff{S'_{i}}{X}=\heff{S_i}{X}$.

Using Eq.\ \eqref{eq:difference-Heff} and recalling the expression for $\nabla \FMED$ in Eq.\
\eqref{eq:gradient-MED}, we obtain
\begin{align}
\|(\nabla \FMED)_{S_i'} & -(-\hefftot{[i+1]}+\hefftot{[i]})\| \nonumber \\
&=\|h_i+\log(\rho_{S_i'})-\log(\rho_{S_i})-(-\hefftot{[i+1]}+\hefftot{[i]})\| \nonumber \\
&=\left\|\log(\rho_{S_i'})-\log(\rho_{S_i})-(-h_i-\hefftot{S_i'}+\hefftot{S_i})\right\| . \label{eq:gradMED-diff}
\end{align}
Now, note that $h_{i} = H_{[i+1]} - H_{[i]} = \sum_{X \subseteq [i+1] : (i+1) \in X} h_{X} = H_{S'_i} - H_{S_i}$ because $\ellshield\ge r$. As a result, 
\begin{align}
\|h_i &+ \log(\rho_{S_i'}) - \log(\rho_{S_i}) -(-\hefftot{[i+1]}+\hefftot{[i]}) \| \nonumber \\
&\le\left\|\log(\rho_{S_i'})+H_{S_i'}+\hefftot{S_i'}+\log(\Tr[e^{-H}])\right\|+\left\|-\log(\rho_{S_i})-H_{S_i}-\hefftot{S_i}-\log(\Tr[e^{-H}])\right\| \nonumber \\
&\le2\eps\,, \label{eq:diff-Heff-S-eps}
\end{align}
where we used property (i) in the last step. Then Eqs.\ \eqref{eq:gradMED-diff} and \eqref{eq:diff-Heff-S-eps} imply
\begin{align}
\|(\nabla &\FMED)_{S_i'}-(-\hefftot{[i+1]}+\hefftot{[i]}) \| \le 2\eps .
\end{align}
Separately, for the $S_0'$-component
\begin{align*}
\|(\nabla\FMED)_{S_0'}&+\log(\Tr[e^{-H}])-\id+\hefftot{[1]}\|\\
&=\|h_0+\log(\rho_{[1]})+\log(\Tr[e^{-H}])+\hefftot{[1]}\|\le\eps\,.
\end{align*}
Therefore, we are done if we can construct $Z$ such that
\begin{equation}\label{eq:iterativeCondPrev}
Z_{S_i'} = \hefftot{[i+1]}-\hefftot{[i]}\,.
\end{equation}
Towards this end, we will construct a sequence of vectors $Z_k$, $k=0,...,N-1$, where $Z_{0}$ will satisfy Eq.\ \eqref{eq:propertyZ k=0}, and $Z_{k}$ will satisfy~Eq.\ \eqref{eq:propertyZ k=0} and~Eq.\ \eqref{eq:propertyZ k > 0} for any $i \leq k$. Then, $Z:=Z_{N-1}$ will satisfy all the desired conditions.

Let us construct $Z_0$ first. Note that for $i=0$ in Eq.~\eqref{eq:diff-Heff-S-eps}, we have
\[
\|h_0+\log(\rho_{[1]})+\log(\Tr[e^{-H}]) + \hefftot{[1]}\|\le\eps\, . 
\]
Thus, to approximately cancel the first term of $\nabla \FMED$, we need to add the following to the first element of the vector
\begin{align*}
(Z_0)_{S_0'}&=\hefftot{[1]}\, .
\end{align*}
We also need to subtract the above from a subsequent elements which we choose as
\begin{align*}
(Z_0)_{S_j'}&=-\sum_{X\cap\{1\}=\{1\}}\heff{[1]}{X}\\
(Z_0)_{S_{N-1}'}&=-\sum_{X\cap\{1\}=\emptyset}\heff{[1]}{X}\,
\end{align*}
for some $j > 0$, so that we satisfy Eq.~\eqref{eq:kerConstruction} and thereby ensure  $Z_0\in\ker(B)^\perp$.
Here we make the choice that $j$ is defined to be the \emph{last possible} index such that the term is included in the respective support, i.e., $j$ is such that $\{1\}\subset S_j'$ and $ (\forall l > j : \{1\}\not\subset S_l')$. All the remaining elements of $Z_0$ are set to equal $0$.

We continue the construction in an analogous fashion for $k>0$. 
For every $k$, we need to add in the $\hefftot{[k+1]}-\hefftot{[k]}$ which approximately cancels the $S'_k$-th component of $\nabla \FMED$. Moreover, we have to compensate for terms that have been added in a previous step.
We will confirm later that these terms are given by
\begin{equation}\label{eq:iterativeCond}
(Z_k)_{S_i'}=-\sum_{X : X\cap[k+1]\subset S_i' \text{ and } \forall j > i:X\cap[k+1]\not\subset S_j'}\heff{[k+1]}{X} \qquad \forall i>k\, .
\end{equation}
Assuming Eq.~\eqref{eq:iterativeCond} holds from the previous iteration for $Z_{k-1}$, we add the following to the $S_k'$ element of the vector $Z_k-Z_{k-1}$:
\begin{align}\label{eq:ZAddition}
(Z_k-Z_{k-1})_{S_k'}&=\hefftot{[k+1]}-\hefftot{[k]}+\sum_{X : X\cap[k]\subset S_k' \text{ and } (\forall j > k : X\cap[k]\not\subset S_j')}\heff{[k]}{X}\,,
\end{align}
which by Eq.~\eqref{eq:iterativeCond} gives
\[
(Z_k)_{S_k'}=\hefftot{[k+1]}-\hefftot{[k]}
\]
as desired. Before defining the subsequent elements of $Z_k-Z_{k-1}$ for the cancellation, let us note that the terms in Eq.~\eqref{eq:ZAddition} are supported in $S_k'$. The last term is trivially satisfied by definition. The argument for the remaining terms goes as follows. By property (iv), $\heff{[k+1]}{X}\neq\heff{[k]}{X}$ only if $d(X,\{k+1\})<\ell$.
Combined with (ii), we conclude that any nonzero terms $\heff{[k+1]}{X}$, $\heff{[k]}{X}$ must be for subsystem $X$ that satisfies $X\subset B_{2\ell}(\{k+1\})$ and so their supports satisfy $X\cap[k+1],X 
\cap[k] \subseteq B_{2\ell}(\{k+1\})\cap[k+1]\subseteq S_k'$.

By the construction in Eq.~\eqref{eq:kerConstruction}, each of the terms in Eq.~\eqref{eq:ZAddition} needs to be subtracted from the other subsequent entries of $Z_k-Z_{k-1}$.
In other words, we have to add the following terms to the remaining entries of $Z_k-Z_{k-1}$:
\begin{align*}
-\sum_{X} \heff{[k+1]}{X}+\sum_{X}\heff{[k]}{X}-\sum_{X : X\subset S_k'\text{ and }(\forall j>k: X\not\subset S_j')}\heff{[k]}{X} \, .
\end{align*}
We treat each $X$ separately (recalling the linearity of $\ker(B)^\perp$). They can be classified into the following cases.
\begin{itemize}
    \item Case 1: $d({k+1},\overline{[k+1]})\ge3\ell$, i.e., the site $k+1$ is surrounded by only previous sites allowing us to distinguish between $X$ close and far from the new site.
    \begin{itemize}
        \item Case 1.1: $X\subset B_\ell(\{k+1\})$. 

        Then $d(X,\overline{[k+1]})\ge d(k+1,\overline{[k+1]})-\ell\ge2\ell$ so by (iii), $\heff{[k+1]}{X}=0$. 
        The other two terms cancel if the condition in the sum holds. If it does not hold, we add $\heff{[k]}{X}$ to $(Z_k-Z_{k-1})_{S_j'}$ for the largest $j>k$ such that $X\cap[k]\subset S_j'$, which exists by the fact that the support condition is violated and $X\cap[k]\subset S_k'$.
        
        \item Case 1.2: $X\cap B_\ell(\{k+1\})=\emptyset$.

        Then also $d(X,\overline{[k+1]})\le \ell$ by (iii) and so there is $j>k+1$, $a\in X$ with $d(a,j)\le \ell$. Therefore, using (ii) $d(a',j)\le 2\ell$ $\forall a'\in X$ and so $X\cap[k+1]\subset S_{j-1}'$. Therefore, the condition in the sum does not hold and we add $\heff{[k]}{X}-\heff{[k+1]}{X}$ to $(Z_k-Z_{k-1})_{S_{j-1}'}$ for the largest such $j$. Note that this updates the terms in the last sum in the iterative condition~Eq.~\eqref{eq:iterativeCond}.
        \item There is no other case because by (iii) $X$ is disjoint from $B_{2\ell}(\{k+1\})\setminus B_\ell(\{k+1\})$ and by (ii) and a triangle inequality $X$ can not include sites more than $2\ell$ and less than $\ell$ sites from $k+1$ simultaneously.
    \end{itemize}

\item Case 2: If case 1 does not apply, then $X\cap[k+1]\subset S_j'$ for some $j>k$.
This can be seen as follows: By (iii) $\exists j\in \overline{[k]}$, $l\in X$, such that $d(j,l)\le\ell$ and by the case distinction there is $j'\in\overline{[k+1]}$, such that $d(j',j)\le3\ell$.
Now since $\diam(X)\le\ell$ by triangle inequality,  $\forall l'\in X$, we have $d(l', j')\le5\ell$ and so $X\cap[k+1]\subset B_{5\ell}(j')\cap[j']\subset S_{j'-1}'$.
Therefore, the second part of the condition in the sum is not fulfilled and we are only left with adding $-\heff{[k+1]}{X}$, $\heff{[k]}{X}$ to $(Z_k-Z_{k-1})_{S_j'}$ for the largest $j$ such that $X\cap[k+1]\subset S_j'$ or $X\cap[k]\subset S_j'$ respectively. Again this updates the terms in Eq.\ \eqref{eq:iterativeCond}.
\end{itemize}
It only remains to observe that as the cancellation does not affect the components $i < k$, we keep for $i < k$, $(Z_k)_{S'_i} = \hefftot{[i+1]} - \hefftot{[i]}$. This shows that $Z_k$ satisfies all the desired properties, namely Eqs.~\eqref{eq:iterativeCondPrev} and~\eqref{eq:iterativeCond}.

Note that for the last term $Z_{N-1}-Z_{N-2}$ no subtraction from subsequent terms is possible but the terms already cancel exactly since we can assume wlog that $\hefftot{[N]}=0$ and we are always in case 1.1.

To conclude the proof we choose $Z=Z_{N-1}$.
\end{proof}

\subsection{Main theorem}

With this we have the ingredients to prove the following theorem on the approximation error of the MED relaxation. 
\begin{theorem}\label{thm:errorEff}
Given an $(\eps,\ell)$-effective interaction and for $\ell_S=5\ell$ we have
\[
|\FMED-\F|\le2N\eps 
\]
Furthermore, in 1D, the marginals $\sigma_{S'_j}$ returned by the MED are close to the marginals $\rho_{S'_j}$ of the Gibbs state:
\begin{equation*}
\|\sigma_{S_j'}-\rho_{S_j'}\|_1\le\sqrt{4N\eps}
\end{equation*}
for all $j$. Moreover, beyond the 1D setting, we always have that $\| \sigma_{1} - \rho_{1} \|_{1} \leq \sqrt{4N\eps}$. For an order such that $[j+1]=B_{\ell_S/2}(\{i\})$ for some site $i$, we have
\begin{equation*}
\|\sigma_{S_j'}-\rho_{S_j'}\|_1\le\sqrt{4N\eps}.
\end{equation*}
\end{theorem}
\begin{remark}
\label{rem:compute-marginal}
    Note that we can use the last part of the statement to compute an approximation of any desired marginal $T$. In fact, we first take any site $i$ such that $T \subset B_{\ell_S/2}(\{i\})$. Our ordering then starts with all the sites  $l \in B_{\ell_S/2}(\{i\})$ with $l \neq i$ and then site $i$ and then the sites outside of $B_{\ell_S/2}(\{i\})$. Then we have $[|B_{\ell_S/2}(\{i\})|] = B_{\ell_S/2}(\{i\})$.
\end{remark}
\begin{proof}
Let $\{\sigma_{S_k'}\}$ be 
For the free energy, we use Lemma~\ref{lem:gradient} to bound  $\delta$ in Lemma~\ref{lem:MEDerror}.
\begin{align*}
\delta&=-\log(\Tr\left[e^{-H}\right])-\left\langle\nabla\FMED(\{\rho_{S_k'}\}),\{\sigma_{S_k'}\}\right\rangle+1\\
&=-\log(\Tr[e^{-H}])-\left\langle\nabla\FMED(\{\rho_{S_k'}\})+Z,\{\sigma_{S_k'}\}\right\rangle+1\\
&\le2N\eps
\end{align*}
with $Z$ from Lemma~\ref{lem:gradient} and using $Z\in\ker(B)^\perp$, $\|AB\|_1\le\|A\|\|B\|_1$, and $\|\sigma_{S_k'}\|_1=1$.

For the marginals, let us first consider the 1D case and assume a consecutive ordering of sites.
In particular, we can use that $S_{k+1}\subset S_k'$.
Recall from Lemma~\ref{lem:MEDerror}
\begin{align*}
\sum_{k=0}^{N-1} D(\sigma_{S_k'}\|\rho_{S_k'})-D(\sigma_{S_k}\|\rho_{S_k})\le\delta\,.
\end{align*}
Notice that by data-processing, we have the inequalities
$D(\sigma_{S_k'}\|\rho_{S_k'})-D(\sigma_{S_k}\|\rho_{S_k})\ge0$ and\linebreak $D(\sigma_{S_k'}\|\rho_{S_k'})-D(\sigma_{S_{k+1}}\|\rho_{S_{k+1}})\ge0$.
Reordering the inequality, we can bound the marginal errors
\begin{align*}
    D(\sigma_{S_j'}\|\rho_{S_j'})&\le\delta-\left(\sum_{k=j+1}^{N-1}D(\sigma_{S_k'}\|\rho_{S_k'})-D(\sigma_{S_k}\|\rho_{S_k})\right)-\left(\sum_{k=0}^{j-1} D(\sigma_{S_k'}\|\rho_{S_k'})-D(\sigma_{S_{k+1}}\|\rho_{S_{k+1}})\right)\\
    &\le\delta\,,
\end{align*}
where we omitted the divergence on an empty set $D(\sigma_{S_0}\|\rho_{S_0})=0$. Thus, by Pinsker's inequality
\begin{equation*}
    \|\sigma_{S_j'}-\rho_{S_j'}\|_1\le\sqrt{2\delta}.
\end{equation*}
Beyond 1D, if $[k+1]=B_{\ell_S/2}(\{j\})$, then the distance between any two sites in $[k+1]$ is at most $\ell_{S}$. As a result, $[k+1] \subset B_{\ell_S}(\{i\})$ for $i \leq k+1$ and the shields are simply $S_{i} = S'_{i-1} = [(i-1)]$.
Therefore, by Eq. \eqref{eq:errDivDelta},
\begin{align*}
    D(\sigma_{S_j'}\|\rho_{S_j'})&=\sum_{k=0}^j D(\sigma_{S_k'}\|\rho_{S_k'})-D(\sigma_{S_k}\|\rho_{S_k})
    \\&\le\delta-\left(\sum_{k=j+1}^{N-1}D(\sigma_{S_k'}\|\rho_{S_k'})-D(\sigma_{S_k}\|\rho_{S_k})\right)\\
    &\le\delta\, .
\end{align*}
\end{proof}

Putting the above together with a suitable assumption on the decay of the effective interaction, we get the following efficient algorithms given by optimizing the MED relaxation.

\begin{corollary}
\label{cor:free_energy_algorithm}
Given a local Hamiltonian on a $D$-dimensional lattice of $N$ sites with an exponentially-decaying effective interaction, there exists an algorithm taking as input a description of the terms $h_X$ and an inverse error $\eps$ and outputting an approximation $\widetilde{\textup{F}}_\textup{MED}$ of the free energy such that
\[
|\F-\widetilde{\textup{F}}_\textup{MED}|\le\eps \, ,
\]
and furthermore outputting marginals $\sigma_A$ on subsets $A$ of constant diameter approximating the thermal state as
\[
\|\sigma_A-\rho_A\|\le\eps
\]
in time
\[
\poly\left(N, \exp(\log(\frac N\eps)^D)\right).
\]
\end{corollary}
\begin{proof}
We use an $\eps/2$ argument to combine the error estimate from Theorem~\ref{thm:errorEff} with the error from the convex optimization itself.
From Theorem~\ref{thm:errorEff}, we need an effective interaction with error $\eps_\text{eff}=\eps/4N$. 
We obtain from Definition~\ref{def:heffexp} for constants $C,\xi$ an
$(\eps_\text{eff},\ell)$-effective interaction if
$\ell_S=5\ell\ge5\xi\log(NC/\eps_\text{eff})=5\xi\log(4N^2C/\eps)$.
This ensures $|\FMED-\F|\le\eps/2$, and this holds for any valid choice of ordering. To obtain guarantees on the marginals, we use Remark~\ref{rem:compute-marginal}, which ensures that for any fixed marginal $A$, we can choose an ordering such that the estimates $\sigma_{A}$ that the MED outputs satisfy $\|\sigma_{A} - \rho_{A} \|_{1} \leq \sqrt{2 \eps}$. We can simply run the MED algorithm with an appropriate ordering for each set $A$.

The optimization of the functional $\FMED(\{\sigma_{S_k'}\})$ itself is a convex program, over $N$ positive semidefinite matrices of size up to $d^{\mathcal O(\ell_S^D)}$.
We omit the details of the convex optimization, but refer to \cite[Theorem (4.3.13)]{Grtschel1993} proving a polynomial-time oracle algorithm to give an $\eps/2$-approximation $\widetilde{F}_\text{MED}$ of $\FMED$.
We note a subtlety regarding the bit complexity: It is difficult to obtain general lower bounds on the minimal eigenvalues of marginals of the Gibbs state, which is needed to ensure accurate approximations of the logarithm.
However, for the purpose of theoretically proving a (quasi-)polynomial runtime it will be sufficient to consider lower bounds on the global Gibbs state, which are always exponentially small in $N$:
\begin{equation*}
\lambda_\text{min}(\rho_{S_k'})\ge \lambda_\text{min}(\rho)\ge e^{-2\|H\|-\log(d)}\ge e^{-J N} \, ,
\end{equation*}
where $J$ can be bounded as $J\le\log(d)+\max_X\|h_X\| 2^{(2r+1)^D}$.
We will impose the bound $\rho_{S_k'}\ge e^{-JN}$ explicitly in our optimization in order to ensure good approximations to the matrix logarithms involved:
\begin{equation*}
\min_{\{\sigma_{S_k'}\} \textrm{ s.t. }B\left(\{\sigma_{S_k'}\}\right)=0, \sigma_{S_k'}> e^{-JN}/2} \FMED\left(\{\sigma_{S_k'}\}\right) \, .
\end{equation*}
This is still a relaxation since a feasible point for which the MED functional lower bounds the free energy is given by the marginals of the Gibbs state.

Now we can use the following oracles (see \cite{Grtschel1993} for details) with polynomial bit complexity:
For the description of the convex set we need a weak membership oracle, which is given in \cite[Theorem 10.10]{Gupta2011}.
For the functional we use an approximation of the logarithm based on a Gaussian representation of the integral representation of the logarithm \cite[Proposition 1]{Fawzi2019}. Here, from the above additional constraint, we can assume that the argument of the logarithm is lower bounded by $e^{-JN}/2$. Under this assumption, the approximation requires a degree $\poly(N)\log(1/\eps)$, which guarantees a polynomial-time oracle for any required error given with polynomially sized bit precision.

The subroutines of the above oracles all run in polynomial-time. Together with the (quasi-)polynomial size of the matrix variables we conclude the runtime claimed in the theorem.
\end{proof}

\begin{remark}
The above proof, which is mostly for theoretical purposes requires a polynomial scaling of the bit complexity. This is not necessary if a better lower bound on the lowest eigenvalues for marginals of the Gibbs state is known to hold, such as exponential decay in the number of sites of the marginal rather than the overall system size. In 1D such bounds are available \cite[Lemma 4.2]{Fawzi2023}. In practical implementations using interior point methods, the lower bounds might not be enforced explicitly. 
\end{remark}

\section{Decay of the effective interaction}\label{sec:efficiency}

In this section we discuss settings where the decay of the effective interaction assumption has been shown to hold.

\paragraph{One dimension} In the one-dimensional case, the existence of an exponentially-decaying effective interaction was essentially proven in \cite{Kuwahara2024}, but we show how it translates to our setup in Appendix~\ref{sec:App1DEff} by proving the following proposition.
\begin{proposition}\label{prop:1DEff}
In one dimension, for the Gibbs state of a finite range interaction and at arbitrary constant temperature, there exists an exponentially decaying effective interaction.
\end{proposition}
\begin{remark}
    Note that this result holds at arbitrary temperature but that the decay rate depends on temperature.
    An extension to the computationally provably hard ground-state problem is therefore not possible. While included in the above proposition, in the commuting setting the stronger result in Example~\ref{example:commuting1D} shows a constant locality, independent of temperature for the exact effective Hamiltonian.
\end{remark}

\paragraph{Higher dimensions} Beyond 1D, we can make use of cluster expansions in the high-temperature regime. While a general theory for effective interactions of noncommuting Hamiltonians is still an open problem, the following lemma proves its existence under a commutativity assumption on the Hamiltonian terms and its marginals based on \cite{bluhm2024}. This includes Hamiltonians constructed from a set of geometrically local generators of CSS codes. We show how the result of \cite{bluhm2024} translates to our setup in Appendix~\ref{sec:appendixHighT} by proving the following proposition.
\begin{proposition}\label{prop:highDEff}
    In any dimension, for the Gibbs state of a finite-range interaction with commuting marginals $[h_X,h_Y]=[\Tr_R[h_{X_1}\ldots h_{X_n}],\Tr_S[h_{Y_1}\ldots h_{Y_m}]]=0$ for all $X,Y,X_1,\ldots, X_n,Y_1,\ldots Y_m,\allowbreak R,\allowbreak S$, and for $\sup_X\|h_X\|=\beta < \beta_*$, where $\beta_*$ is a critical temperature only depending on the lattice and locality, there exists an exponentially-decaying effective interaction.
\end{proposition}

While the above proofs are just immediate consequences of prior work, by proving the conditions in Section~\ref{sec:general_theory}, we find the following runtimes for the MED algorithm as of Corollary~\ref{cor:free_energy_algorithm}.

\begin{corollary}
Given as input $\eps>0$  and a local Hamiltonian on a lattice either fulfilling
\begin{itemize}
    \item D=1, or
    \item commutativity $[\Tr_R[h_{X_1}\ldots h_{X_n}],\Tr_S[h_{Y_1}\ldots h_{Y_m}]]=0$ for all $X_1,\ldots, X_n,Y_1,\ldots Y_m, R,S$, and $\sup_X\|h_X\|=\beta < \beta_*$, where $\beta_*$ is a critical temperature only depending on the lattice and locality,
\end{itemize}
then we have that the MED algorithm gives approximations of the free energy $\widetilde\F_\textrm{MED}$ and marginals $\sigma_A$ on sets of constant diameter
\begin{align*}
    |\F-\widetilde\F_\textrm{MED}|&\le\eps \, ,\\
    \|\sigma_A-\rho_A\|_1&\le\eps \, ,
\end{align*}
in time
\[
\poly\left(N, \exp(\log(\frac N\eps)^D)\right).
\]
\end{corollary}

\section{Rounding} \label{sec:rounding}
While we discussed convergence guarantees for approximations to marginals of the Gibbs states, these approximations are not necessarily compatible with a global state.
In this section, we discuss the problem of rounding this relaxed solution to a global state.
The construction is given by a sequence of recovery channels.
We present two consequences of the scheme.

Firstly, the implementation of the channels allows for the efficient preparation of all states fulfilling the requirements on a quantum computer.
In fact these requirements, namely the decay of conditional mutual information and access to approximations of marginals, are a subset of the convergence criteria of the MED.
They guarantee efficient sampling for a potentially wider class of states than those presented in Section~\ref{sec:efficiency}.
In fact, in some models such as the toric code~\cite{castelnovo2007entanglement}, the CMI decay can be shown even at low temperature\footnote{More precisely, CMI decay can be guaranteed provided that the size of the subsystem is larger compared to $e^{\beta}$, where $\beta$ is the inverse temperature.}.
For such systems, any method that yields an accurate approximation of the marginals can be used in our scheme.
Secondly, we also explain how the construction provides a tensor network description of the underlying states.

\subsection{Recovery Channels}

In the previous sections we made use of an effective interaction, a strong notion of decay of correlations that is related to an operator version of the conditional mutual information \cite{capel2024}.
In this section, we will only use the related and more established conditional mutual information itself.
Small or vanishing CMI is known to imply various structural properties of the state including its blockwise factorization \cite{Hayden_2004,salzmann2024} and a representation using recovery channels \cite{petz1988,fawzi_quantum_2015,junge2018}.

In particular, we will make use of the latter results. \cite{junge2018} gave an explicit expression for an approximate recovery channel for states with small CMI, whose existence had previously been proven in \cite{fawzi_quantum_2015}.
We introduce the following definitions and results from there as our construction is heavily based on those (the reference deals with more general cases of the data-processing inequality for the quantum relative entropy).
Consider a state $\rho$ on a tripartite system $ABC$. The rotated Petz map is defined as
\begin{align*}
\beta_0(t)&=\frac\pi2(\cosh(\pi t)+1)^{-1} \, ,\\
\mathcal R^t_{\rho_{BC},\Tr_C}(X)&=\rho_{BC}^{1/2-it}\rho_B^{it-1/2}X\rho_B^{-it-1/2}\rho_{BC}^{1/2+it} \, ,\\
\mathcal R_{\rho_{BC},\Tr_C}(X)&=\int_{\mathbb R}\beta_0(t)\mathcal R_{\rho_{BC},\Tr_C}^{t/2}(X) \, .
\end{align*}
The main result regarding the approximate recovery is an upper bound on the recovery error in terms of the CMI.
\begin{theorem}[{\cite[Theorem 2.1, Remark 2.2]{junge2018}}]
The following inequalities hold:
\begin{align}\label{eq:CMItoF}
I(A:C|B)_\rho&\ge-2\log (F(\rho_{ABC},\mathcal R_{\rho_{BC},\Tr_C}(\rho_{AB}))) \geq \frac{1}{4} \left\| \rho_{ABC} - \mathcal R_{\rho_{BC},\Tr_C}(\rho_{AB})) \right\|^2_1
\end{align}
where the fidelity $F(\rho, \sigma) = \| \sqrt{\rho} \sqrt{\sigma} \|_1$ and $\| \cdot \|_1$ is the trace norm.
\end{theorem}

\subsection{Rounding via Channel Concatenation}
We still consider the same setup as in the MED with an ordered set of sites $1,\ldots,N$ and Markov Shields $S_k=B_\ell(k+1)\cap[k]$.
In addition, we will need a doubled Markov shield $S^2_k=B_{2\ell}(k+1)\cap[k]$, $S'^2_k=S^2_k\cup\{k+1\}$.

The idea of the following theorem is to construct a state from a sequence of recovery channels site by site.
We require the channel to act on a shield of radius $\ell$ and to recover the approximate marginals well on a shield of radius $2\ell$, but not globally so that the approach remains computationally feasible.
To prove a good global approximation, we leverage the existence of a second channel reconstructing the original global state from the enlarged shield, which commutes with the constructive channel.
This latter channel cannot be efficiently constructed, however, it is only part of the proof, not the construction, so no computational access to it is needed.

\begin{theorem}\label{thm:rounding}
Let $\rho$ be a state such that $\forall$ disjoint $A,B,C\subset\Lambda$ with $d(A,C)\ge\ell$, $I(A:C|B)\le\eps_{CMI}$. 
For a set of marginals $\sigma_{S'^2_k}$ such that $\|\sigma_{S'^2_k}-\rho_{S'^2_k}\|_1\le\eps_\sigma$, we define a rounding scheme.
Let
\begin{align*}
\phi_{k+1}&: \mathcal{B}(\mathcal H_{S_k})\to\mathcal{B}(\mathcal H_{S_k'})\\
\phi_{k+1}&=\mathcal R_{\sigma_{S_k'},\Tr_{k+1}}.
\end{align*}
Then the rounded state is given as
\[
\widetilde\sigma=\bigcirc_{i=2}^N(\mathrm{id}_{\overline{S_{i-1}}} \otimes \phi_i)(\sigma_{\{1\}})
\]
and fulfills
\[
\|\widetilde\sigma-\rho\|_1\le N(4\sqrt{\eps_{CMI}}+2\eps_\sigma+2\sqrt{\eps_{CMI}+3\log d\sqrt\eps_\sigma}).
\]
The channels and thereby the circuit constructing $\widetilde\sigma$, are given explicitly in terms of the marginals $\sigma_{S'_k}$.
\end{theorem}
\begin{proof}
We will leverage the small conditional mutual information to construct good recovery channels for both, the global state $\rho$ as well as, using an appropriate continuity bound, for the approximate marginals $\sigma$. By  \cite[Lemma 4]{Shirokov-ContinuityBounds-2019}, we have
\begin{align*}
  I(k+1:S^2_k\setminus S_k|S_k)_{\sigma_{S'^2_k}}& \le I(k+1:S^2_k\setminus S_k|S_k)_{\rho_{S'^2_k}}+2\log d \, \frac{\eps_\sigma}{2} + \left(1+\frac{\eps_\sigma}{2}\right) h\left(\frac{1}{1+\eps_\sigma/2}\right)\\
  & \leq \eps_{CMI} + 2  \log d \sqrt{\frac{\eps_\sigma}{2}} + \sqrt{2 \frac{\eps_\sigma}{2}} \\
  & \leq \eps_{CMI} + 
  3\log d\sqrt{\eps_\sigma} \, ,
\end{align*}
where we are using that $\eps_\sigma/2 \in (0,1)$, and thus ${ \frac{\eps_\sigma}{2}}\leq  \sqrt{ \frac{\eps_\sigma}{2}} $, $h(x)=-x\log x -(1-x) \log(1-x)$ is the binary entropy, and that $(1+x)h\left( \frac{1}{1+x}\right) \leq \sqrt{2x}$. 
Moreover, by Eq. \eqref{eq:CMItoF},
\[
\|\phi_{k+1}(\sigma_{S^2_k})-\sigma_{S'^2_k}\|_1\le2\sqrt{\eps_{CMI}+3\log d\sqrt{\eps_\sigma}} \, .
\]
Since the error bound for this recovery channel only applies for the small subsystem $S'^2_k$, we will in addition make use of the global CMI condition which only applies to $\rho$ (since the approximate marginals are not globally defined/might not be consistent):
\[
\mathcal{R}_{k+1}=\mathcal R_{\rho_{[k]\setminus S_k},\Tr_{[k]\setminus S^2_k}} \, ,
\]
which acts on the outer part of the extended Markov shield $S^2_k\setminus S_k$ and recovers all preceding sites $[k+1]\setminus S'^2_k$.
Using the Markov property of $\rho$, we have
\[
\|\mathcal R_{k+1}(\rho_{S'^2_k})-\rho_{[k+1]}\|_1\le2\sqrt{\eps_{CMI}} \, .
\]
We define the rounding iteratively as 
\begin{align*}
\widetilde\sigma_1&=\sigma_{\{1\}}\\
\widetilde\sigma_{k+1}&=\phi_{k+1}(\widetilde\sigma_k) \, .
\end{align*}
Finally, using that the trace-norm is contractive under quantum channels, we have
\begin{align*}
\|\widetilde\sigma_{k+1}-\rho_{[k+1]}\|_1& =\|\phi_{k+1}(\widetilde\sigma_{k})-\rho_{[k+1]}\|_1\\
&\le\|\phi_{k+1}(\widetilde\sigma_k)-\phi_{k+1}(\rho_{[k]})\|_1\\
&\quad+\|\phi_{k+1}(\rho_{[k]})-\phi_{k+1}(\mathcal R_{k+1}(\rho_{S^2_k}))\|_1\\
&\quad+\|\phi_{k+1}(\mathcal R_{k+1}(\rho_{S^2_k}))-\phi_{k+1}(\mathcal R_{k+1}(\sigma_{S^2_k}))\|_1\\
&\quad+\|\phi_{k+1}(\mathcal R_{k+1}(\sigma_{S^2_k}))-\mathcal R_{k+1}(\sigma_{S'^2_k})\|_1\\
&\quad+\|\mathcal R_{k+1}(\sigma_{S'^2_k})-\mathcal R_{k+1}(\rho_{S'^2_k})\|_1\\
&\quad+\|\mathcal R_{k+1}(\rho_{S'^2_k})-\rho_{[k+1]}\|_1\\
&\le\|\widetilde\sigma_k-\rho_{[k]}\|_1\\
&\quad+\|\rho_{[k]}-\mathcal R_{k+1}(\rho_{S^2_k})\|_1\\
&\quad+\|\rho_{S^2_k}-\sigma_{S^2_k}\|_1\\
&\quad+\|\phi_{k+1}(\sigma_{S^2_k})-\sigma_{S'^2_k}\|_1\\
&\quad+\|\sigma_{S'^2_k}-\rho_{S'^2_k}\|_1\\
&\quad+\|\mathcal R_{k+1}(\rho_{S'^2_k})-\rho_{[k+1]}\|_1\\
&\le\|\widetilde\sigma_{k}-\rho_{[k]}\|_1+4\sqrt{\eps_{CMI}}+2\eps_\sigma+2\sqrt{\eps_{CMI}+3\log d\sqrt\eps_\sigma} \, .
\end{align*}
The claim follows by induction.
\end{proof}

Decomposing the channels in the previous theorem into a quantum circuit (of exponential size in the number of qubits it acts on), we obtain the following result on efficient Gibbs sampling for logarithmically sized shields:
\begin{corollary}\label{cor:circuitDecomposed}
Let $\eps>0$.
    Under the conditions of Theorem~\ref{thm:rounding}, where $1/\eps_{CMI},1/\eps_\sigma$ need to be chosen to some accuracy $\poly(1/\eps)$, there is a quantum algorithm that runs in time
    \begin{equation*}
        \poly\left(\exp(\ell^D),N,1/\eps\right)
    \end{equation*}
    and outputs an approximation $\sigma$ of the Gibbs state with
    \[
    \|\sigma-\rho\|_1\le\eps \, .
    \]
\end{corollary}
\begin{proof}
We can use Theorem~\ref{thm:rounding} which gives an explicit formula for $N$ channels constructing an $\eps/2$ approximation to the Gibbs state using $N$ channels (the input marginal can be seen as another channel without input).
We need to give an $\eps/2N$ approximate circuit decomposition in diamond norm for each of these channels, such that the overall output fulfills the required error bound.
Due to the decomposition, each channel acts on/outputs $\mathcal O(\ell^D)$ many qubits and by a Stinespring dilation can be seen as a unitary of that size (adding and tracing out $\mathcal{O}(\ell^D)$ ancilla qubits).

We defer the discussion of the numerical computation of a gate decomposition of the rotated Petz map to Appendix~\ref{sec:AppDetailsRotPetz}, showing that an algorithm with runtime $\exp(\mathcal{O}(n))\poly\log(1/a\eps)/a$ exists that constructs a decomposition into gates with error $\eps/2$ for the rotated Petz map acting on $n$ qubits and defined via marginals with smallest eigenvalue lower bounded by $a$.

However, in order to avoid additional assumptions on these marginals we note the following:
We consider the error $\eps_\sigma=1/\poly(1/\eps)$ needed for an $\eps/2$ approximation in Theorem~\ref{thm:rounding}, find $\eps_\sigma/2$ approximations to the marginals, and apply the construction in the theorem to $\sigma_{S_k'}+\id \eps_\sigma/2\dim\left(\mathcal H_{S_k'}\right)$ such that the marginals' smallest eigenvalues are lower bounded by $a\ge\eps_\sigma/2\dim(\mathcal H_{S_k'})$ and the marginals still fulfill the required error bounds.
\end{proof}
The construction can also be seen from a classical computational perspective. Consider the sequence of channels in the above Corollary~\ref{cor:circuitDecomposed}. We can view the application of each channel as a tensor applied to the respective sites and decompose each of these tensors using singular value decomposition into local tensors only connected to neighbouring sites whose bond dimension is no larger than exponential in the number of sites per channel. Furthermore each site is only acted on by a number of channels scaling as $\mathcal{O}(\ell^D)$. Collecting these tensors for each site we obtain the following corollary.
\begin{corollary}\label{cor:mpo}
    Under the conditions of Corollary~\ref{cor:circuitDecomposed}, there is a projected entangle pair operator $\sigma$ approximating the Gibbs state, $\|\sigma-\rho\|_1\le\eps$ with bond dimension
    \begin{equation*}
    D_\textrm{bond}=\poly\left(\exp(\ell^D),N,1/\eps\right)\,.
    \end{equation*}
\end{corollary}

The results in this section are conditional on the existence of an efficient classical algorithm for marginal approximations and the decay of the CMI. Both follow from the existence of an effective interaction, the former by using the MED algorithm, the latter quite directly by Lemma~\ref{lem:CMI_decay}.
In the settings where we were able to prove this condition, we collect these results and state the following corollary.
Note that alternative classical algorithms and proofs of the CMI decay are possible under weaker conditions than the exponentiual decay of the effective interaction, but we leave exploring such results to future work.

\begin{corollary}
Given as input $\eps>0$  and a local Hamiltonian on a lattice either fulfilling
\begin{itemize}
    \item D=1, or
    \item commutativity $[\Tr_R[h_{X_1}\ldots h_{X_n}],\Tr_S[h_{Y_1}\ldots h_{Y_m}]]=0$ for all $X_1,\ldots, X_n,Y_1,\ldots Y_m, R,S$, and $\sup_X\|h_X\|=\beta < \beta_*$, where $\beta_*$ is a critical temperature only depending on the lattice and locality,
\end{itemize}
then the Gibbs sampling quantum algorithm in Corollary~\ref{cor:circuitDecomposed}
outputs an approximation to the Gibbs state $\sigma$ with
\[
\|\sigma-\rho\|_1\le \eps
\]
in time
\[
\poly\left(N, \exp(\log(\frac N\eps)^D)\right)\,.
\]
In addition, there exists a classical MPO approximation $\sigma_{MPO}$ with $\|\sigma_{MPO}-\rho\|_1\le\eps$ and bond dimension $D_\textrm{bond}=\poly\left(N, \exp(\log(\frac N\eps)^D)\right)$.
\end{corollary}
\begin{proof}
    The proof is simply the concatenation of the results in this paper.
    By Propositions~\ref{prop:1DEff} and~\ref{prop:highDEff}, there exists an exponentially decaying effective interaction as of Definition~\ref{def:heffexp}.
    This also implies the exponential decay of the CMI by Lemma~\ref{lem:CMI_decay}.
    By Corollary~\ref{cor:circuitDecomposed}, we need to ensure sufficiently small $1/\eps_{CMI},1/\eps_\sigma=\poly(1/\eps)$.
    To ensure small $\eps_{CMI}$, using the decay we can choose $\ell=\log(\mathcal{O}(N/\eps_{CMI}))=\log(\mathcal{O}(N/\eps))$.
    We run the MED choosing the Markov shield size $\ell_S\ge\max\{2\ell,\log(\mathcal{O}(N/\eps))\}$, where the first lower bound is a requirement of Theorem~\ref{thm:rounding} on the marginals we need to obtain, and the second ensures the error bound, see Corollary~\ref{cor:free_energy_algorithm}.
    Using the classical algorithm to construct the decomposition into quantum gates as of Corollary~\ref{cor:circuitDecomposed}, the result follows.
    The statement on MPO approximations follows from Corollary~\ref{cor:mpo} using the same requirements.
    
\end{proof}

\subsection{Comparison with prior work}\label{sec:comparisonBrandao}
We compare Theorem \ref{thm:rounding} to the result of \cite{brandao2019}. This work also describes an approach to Gibbs state preparation of lattice models based on the Petz recovery map, but differs in that the inputs of the algorithm are Gibbs states of truncated versions of the global Hamiltonian, which are then acted upon by a constant-depth sequence of Petz maps, to which the authors assume black-box access. To derive their result, they assume two constraints on the Gibbs states of any truncated Hamiltonian $H_{X}$ for $X\subset \Lambda$: exponential decay of correlations and  exponential decay of the CMI for all tripartitions $ABC=X\subset \Lambda$ where $A$ is shielded from $C$ via $B$.

Our construction on the other hand only assumes a system with an exponentially decaying effective interaction, and takes as inputs approximate marginals of the global Gibbs states obtained from the MED as inputs. It then constructs $N$ layers of this recovery map from said marginals and prepares the global Gibbs state. It should be noted that in terms of a decomposition into one and two-qubit gates, our construction has a polynomial (for 1D) or quasi-polynomial (for higher dimensions) depth in $N,1/\eps$, whereas this is unclear for \cite{brandao2019} due their construction requiring the exact marginals of the Gibbs state with no clear description of how to attain them.
Also note that our algorithm does only depend on the Markov structure of the states and therefore is applicable to settings outside of thermal state preparation, whereas \cite{brandao2019} refers to the truncated Hamiltonians and thereby only to the setting of Gibbs states.

To get some additional insight into the issue of constructing the Petz map using approximations to the marginals, we present a continuity bound for the rotated Petz recovery map, which we prove in Appendix~\ref{sec:AppPetzCont}.
\begin{proposition}\label{prop:rotPetzCont}
For quantum states $\rho_{BC}$, $\sigma_{BC}$ lower bounded by a constant $a\id\le\rho_{BC},\sigma_{BC}$, the following perturbation bound for the rotated Petz recovery map holds:
\begin{equation*}
\left\|\mathcal R_{\rho_{BC},\Tr_C}-\mathcal R_{\sigma_{BC}, \Tr_C}\right\|_\diamondsuit\le\frac{6d_{BC}}{a^{5/2}}\|\rho_{BC}-\sigma_{BC}\|
\end{equation*}
\end{proposition}
While we cannot rule out the existence of a result without dimensional dependence as it has been achieved for the standard Petz map \cite[Theorem 3.9]{gao2022}, the rotation in the maps which is not an operator monotone function poses a substantial technical challenge to improvement.
However, this result already is sufficient to provide an efficient construction whose gates are efficiently computable in (quasi-)polynomial time as long as a bound of the form $a\le\exp(\mathcal O(|BC|))$ can be shown. For 1D translation-invariant systems this is known \cite[Lemma 4.2]{Fawzi2023}. While in a high-temperature setting in arbitrary dimensions such a result seems reasonable to expect, we are not aware of any results in this or other settings at this point.

While preparing this manuscript we became aware of the concurrent work \cite{chen2025quantumgibbsstateslocally}. The authors present an algorithm very similar to the one in \cite{brandao2019}, but also provide an explicit circuit description in terms of time-averaged Lindbladian Gibbs samplers. The quality of recovery of these maps is proven directly and implies the decay of CMI, whereas here and in \cite{brandao2019} they are a consequence of the CMI decay. The additional condition on the decay of correlations remains.

\section*{Acknowledgments}
A.N.C. thanks Adam Bene Watts for discussions.
S.O.S. acknowledges support from the UK Engineering and Physical Sciences Research Council (EPSRC) under grant number EP/W524141/1.
A.C. acknowledges the support of the Deutsche Forschungsgemeinschaft (DFG, German Research Foundation) - Project-ID 470903074 - TRR 352.  This project was funded within the QuantERA II Programme which has received funding from the EU’s H2020 research and innovation programme under the GA No 101017733. A.C., A.N.C. and I.K. are grateful for the hospitality of Perimeter Institute where part of this work was carried out. Research at Perimeter Institute is supported in part by the Government of Canada through the Department of
Innovation, Science and Economic Development and by the Province of Ontario through the Ministry of
Colleges and Universities. This research was also supported in part by the Simons Foundation through the
Simons Foundation Emmy Noether Fellows Program at Perimeter Institute. I.K. acknowledge support from NSF under award number PHY-2337931. This work was done in part while I.K., A.C., and A.N.C. were visiting the Simons Institute for the Theory of Computing, supported by NSF QLCI Grant No. 2016245. H.F. was partially funded by UK Research and Innovation (UKRI) under the UK government’s Horizon Europe funding guarantee EP/X032051/1.
A.T. acknowledges support from the Australian Research Council (FT230100653).

\section*{Statements and Declarations}
\paragraph{Data availability} Data sharing is not applicable to this article as no datasets were generated or analyzed during the current study.
\paragraph{Conflict of interest} The authors have no Conflict of interest to declare that are relevant to the content of this article.

\clearpage
\bibliographystyle{myhamsplain2}
\bibliography{main}

\clearpage
\appendix
\section{Existence proof of an exponentially decaying effective interaction in 1D}\label{sec:App1DEff}
The following proof of Proposition~\ref{prop:1DEff} adapts the results from~\cite{Kuwahara2024} to our definitions.
\begin{proof}
We assume without loss of generality a 2-local interaction.
In one dimension this covers the general case by considering a blocking of any finite-range system.
Also note, that in \cite{Kuwahara2024} it is assumed that $\log(\Tr[e^{-H}])=0$ and we make the same assumption below. This can always be achieved by adding a constant to the Hamiltonian terms, but the general case follows, since the definition of the exact and thereby approximate effective interaction, Definition~\ref{def:heffexact} and \ref{def:heff} does not change under addition of a constant.

In the following we will construct an $\ell/2$-local effective interaction for any connected subset $|B|\ge\ell$ with error $\exp(-\Omega(\ell))$.
For subsets $|B|<\ell$, we simply choose the exact effective interaction since it automatically fulfills the locality conditions.
Together, these construct an $(\exp(-\Omega(\ell)),\ell)$-local effective interaction as of Definition~\ref{def:heff}.

The main technical tool of \cite{Kuwahara2024} is the repeated application of the Quantum Belief Propagation, or rather, several versions of it \cite{Hastings2010,Kim2012}.
We extract the relevant results and discuss how they fulfill our definition of an exponentially decaying Hamiltonian.

 \begin{figure}[h!]
\begin{center}

\includegraphics[scale=0.5]{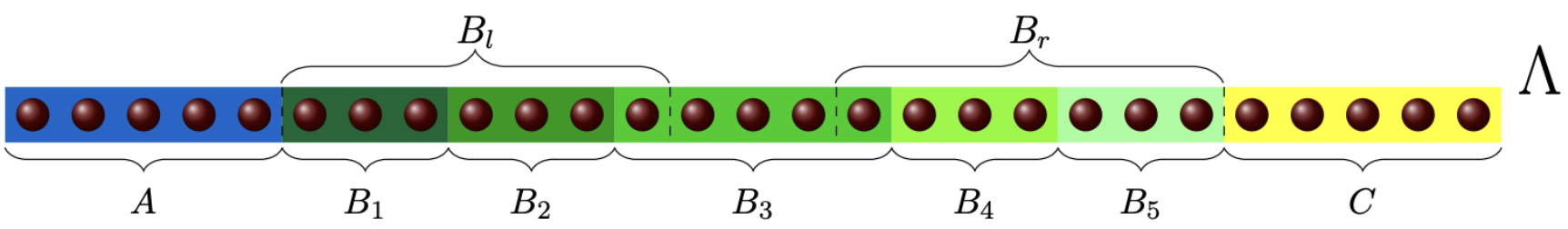}
  \caption{An interval $\Lambda$ split into three subintervals $\Lambda=ABC$ such that $B$ shields $A$ from $C$, where $B$ is further split into five subintervals $B_1,\ldots, B_5$ with $B_1B_2$ and $B_4B_5$ corresponding to the left and right effective interaction terms respectively. }
  \label{fig:1}
 \end{center}
 \end{figure}

We consider the effective interaction on a chain $ABC$, where $B$ is further subdivided into consecutive connected subregions $B_1,\ldots,B_5$ as in Figure \ref{fig:1}. We further denote by $B_l$ the regions $B_1B_2$ together with the succeeding site, i.e., including the support of the interaction with $B_3$, and equivalently $B_r$ for $B_4B_5$ together with the preceeding site.
Note that, in \cite{Kuwahara2024}, thermal states are defined as the exponential of the positive Hamiltonian but we amend the minus sign in line with standard notation.
The construction is given by \cite[(S.717)]{Kuwahara2024} as follows
\begin{align}\label{eq:kuwaharaEff}
\tilde H_{B_l}+H_{B_3}+\tilde H_{B_r}&=U_{B_1B_2}'U_{B_4B_5}'(\tilde H_{B_1}^*+H_{B_2B_3B_4}+\tilde H_{B_5}^*+\hat{\tilde\Phi}'_{B_1B_2}+\hat{\tilde\Phi}'_{B_4B_5})U_{B_4B_5}^{'\dagger}U_{B_1B_2}^{'\dagger}\\
\tilde H_{B_l}&=U_{B_1B_2}'(H_{B_1}^*+H_{B_l\setminus B_1}+\hat{\tilde\Phi}'_{B_1B_2})U_{B_1B_2}^{'\dagger}\\
\tilde H_{B_r}&=U_{B_4B_5}'(H_{B_5}^*+H_{B_r\setminus B_5}+\hat{\tilde\Phi}'_{B_4B_5})U_{B_4B_5}^{'\dagger}
\end{align}
We refrain from giving detailed definitions of all operators involved but recall the following facts:
Each of these operators acts on the subsystems as designated, namely e.g. $U'_{B_1B_2}$ is a unitary acting on $B_1B_2$.
The definition of
\[
\tilde H^*_{B_1}=-\log(\tr_A[e^{-H_{AB_1}}])
\]
does not involve the total size of $B$.
Similarly, the definitions of $U_{B_1B_2}'$ and $\hat{\tilde\Phi}'_{B_1B_2}$ only depend on the choice of $B_1B_2$ and the Hamiltonian terms in that region, and equivalently for the operators acting on $B_4B_5$.

In the 1D setting it is sufficient to define two local terms for the decomposition of the effective interaction. For a given $\ell$, we choose $|B_1|=|B_2|=|B_4|=|B_5|=\lfloor(\ell-1)/4\rfloor$.

The main result from \cite{Kuwahara2024} that we use is equation (S.718), note the missing log-partition function which is assumed to be zero:
\begin{equation}\label{eq:errEff1D}
\|\log(\rho_B)+\tilde H_{B_1B_2}+H_{B_3}+\tilde H_{B_4B_5}\|\le e^{-\Omega(\ell)}
\end{equation}

Then the only non-vanishing terms of the effective interactions can be defined as
\begin{align*}
    \heff{B}{B_l}&=\tilde H_{B_1B_2}-H_{B_l}=U_{B_1B_2}'(\tilde H^*_{B_1}+H_{B_l\setminus B_1}+\hat{\tilde\Phi}'_{B_1B_2})U_{B_1B_2}^{'\dagger}-H_{B_l} \, ,\\
    \heff{B}{B_r}&=\tilde H_{B_4B_5}-H_{B_r}=U_{B_4B_5}'(\tilde H^*_{B_5}+H_{B_r\setminus B_5}+\hat{\tilde\Phi}'_{B_4B_5})U_{B_4B_5}^{'\dagger}-H_{B_r} \, ,
\end{align*}
Now, we can check the items in Definition~\ref{def:heff}.
(i) is equivalent to Eq.~\eqref{eq:errEff1D} with $\eps=e^{-\Omega(\ell)}$, in line with the definition of an exponentially-decaying Hamiltonian.
(ii) and (iii) are fulfilled, since only terms supported on $B_l$ and $B_r$ which are of size $\ell/2$ and adjacent to the boundary are present.
(iv) is a consequence of the definitions of the operators in Eq.~\eqref{eq:kuwaharaEff}, which as explained above do not depend on the size of $B_3$. 
\end{proof}

\section{Existence proof of an exponentially-decaying effective interaction in any dimension, at high temperature, with commuting marginals}\label{sec:appendixHighT}

Here we show the proof of Proposition \ref{prop:highDEff} by adapting the results of \cite{bluhm2024} to our setup.

\begin{proof}
The effective interaction will be essentially equivalent to the strong effective Hamiltonian proposed in \cite[Definition 3.1]{bluhm2024}, whose properties we recall here:
\begin{enumerate}
\item[(a)]  ${\Phi}_{X}^{A}$ is supported in $X \cap A$ for every  $X \subset \Lambda$.
\item[(b)] If $A' \subset \Lambda$, then ${\Phi}^{A}_{X} = {\Phi}^{A'}_{X}$ for all  $X \subset \Lambda$ such that $X \cap A' = X \cap A$.
\item[(c)] We have
\begin{equation}\label{eq:strongEff}
- \log \left( \Tr_{{A}^c}[e^{-H}]\right) = -\log(d^{|A^c|}) +  \sum_{X \subset \Lambda } {\Phi}^{A }_{X}\, 
\end{equation}
\end{enumerate}

As opposed to our approach, the reference does not provide a strictly local approximation but rather a quasi-local exact expression for the true effective Hamiltonian.
We will adapt this definition by simply removing all terms with $\diam(X)\ge\ell$ and use the absolute convergence to prove a tail bound on the truncated terms for sufficiently high temperature.
In addition, since we are interested in the effective interaction rather than the effective Hamiltonian, we will remove the truncated Hamiltonian as well as the dimensional constant in Equation~\eqref{eq:strongEff}.

The definition we make is as follows
\[
\heff{A}{X}=\begin{cases*}
\Phi^A_X& $\diam(X)<\ell$\\
0&$\diam(X)\ge\ell$
\end{cases*}
-\begin{cases*}
\log(d)&if $X=\{i\}\not\in A$\\
0&else
\end{cases*}
-\begin{cases*}
h_X&if $X\subset A$\\
0&else
\end{cases*}\,,
\]
where the second term is for the purpose of the normalization and the third term removes the Hamiltonian terms in the interior.
Let us first observe that Definition~\ref{def:heff} (ii)-(iv) are satisfied by this definition.
For (ii), this is directly incorporated into the definition for the first term by (a) in the definition of the strong effective Hamiltonian. For the second and third term this is immediate as long as $\ell\ge r+1$.

For (iii), note that it was already observed in~\cite{bluhm2024} that $\Phi_X^A=h_X$ if $X\subset\Lambda$, which means the first and last term cancel in this case and the second term does not appear by definition.
Therefore we are only left with terms for which $X\ni i\not\in A$, but this implies (iii) by $\diam(X)\le\ell$.

Additionally, (iv) is immediate from (b) (in fact (b) is stronger as it only requires $X\cap B=\emptyset$ in (iv)).

Let us turn our attention to (i). Comparing Equation~\eqref{eq:strongEff}, Definition~\ref{def:heffexact}, and our new definition, we find that we need to bound the operator norm of 
\[
\heffexact{A}-\hefftot{A}=\sum_{X\subset\Lambda, \ \diam(X)\ge\ell} \Phi_X^A\,.
\]

We consider \cite[Theorem 3.8]{bluhm2024}, where, in the finite-range case that we consider, we choose $\textbf{b}(X)=0$ and assume $[h_X,h_Y]=[\Tr_R[h_X],\Tr_S[h_Y]]=0$ for all $X,Y,R,S$. It states that there is a constant that we denote by $\beta^*$ and that only depends on the lattice and interaction range, such that for every $\sup_X\|h_X\|=\beta\le\beta^*$, we have 
\[
\sup_{x\in V}\sum_{X\ni x}\|\Phi^A_X\| \le 1\,,
\]
and so
\[
\sum_{X}\|\Phi^A_X\| \le |\Lambda|\,.
\]
To obtain a tail bound we start from the weaker bound
\begin{equation}\label{eq:orderBound}
\sum_{X, \ \diam(X)=m}\|\Phi^A_X\| \le |\Lambda|
\end{equation}
and consider the scaling of the terms at fixed $m$ with the inverse temperature.

To that end we need to understand how the set $X$ connects with the terms and orders in the cluster expansion.
While the above bounds already ensure convergence, we want to show in addition that the terms in~\eqref{eq:orderBound} scale at least with $\beta^m$ to deduce a tail bound.
Since we are only interested in the scaling, in the following we will not look at a complete expression including combinatorial factors, but only show that each term in \eqref{eq:orderBound} is a sum of positive terms of order $m$ or higher in the Hamiltonian terms.
We can then obtain a tail bound by a rescaling of~\eqref{eq:orderBound}

We consider equation (42) in~\cite{bluhm2024} in the setting of Theorem 3.10
\begin{align}
\Phi_X^A &=\sum_{l=1}^\infty \sum_{\substack{(\gamma_1,\ldots,\gamma_l)\in\mathbb P^m\\ \gamma_1\lor\ldots\lor\gamma_l=X}}\phi(\gamma_1,\ldots,\gamma_l)\Pi_{j=1}^l\textbf{w}_\beta(\gamma_j)\label{eq:ursellDecomp}\\
\textbf{w}_\beta(\gamma)&=\frac{1}{d_{L^c}}\tr_{L^c}\left[\prod_{Y\in\gamma} h_Y\right]\nonumber
\end{align}
We refer to~\cite{bluhm2024} for the details of the definitions above, but point out the following facts: The $\gamma_i$ called polymers are multisets of support sets $X$ of Hamiltonian terms with the additional constraint of being connected. The set of such connected polymers is denoted by $\mathbb P$.
The Ursell function $\phi$ vanishes if the union $\gamma_1\ldots\gamma_m$ does not form a connected cluster and the norm bound~\eqref{eq:orderBound} is derived by bounding the norms of $\phi$ and $\textbf{w}_\beta$, 
\begin{equation}\label{eq:orderBoundButchered}
\sum_{X, \ \diam(X)=m}\|\Phi^A_X\| \le\sum_{X, \ \diam(X)=m}\sum_{l=1}^\infty \sum_{\substack{(\gamma_1,\ldots,\gamma_l)\in\mathbb P^l\\ \gamma_1\lor\ldots\lor\gamma_l=X}}|\phi(\gamma_1,\ldots,\gamma_l)|\Pi_{j=1}^l\|\textbf{w}_\beta(\gamma_j)\|\le |\Lambda|
\end{equation}
see \cite[Proof of Theorem 3.10]{bluhm2024}.
We are then interested in the minimum number of factors $h_X$ in each element of the sum above.
 
Now for a given term in the sum~\eqref{eq:ursellDecomp} consider an enumeration $X_1,\ldots, X_n$ of all $X$ in the $\gamma_1,\ldots,\gamma_m$.
We only have contributions for $\cup_{i=1}^n X_i$, where $X=(X_1,\ldots,X_n)$ is a \emph{connected cluster} meaning that there is no permutation on $n$ sites $\pi$ and $1\le k< n$ such that $X_{\pi(i)}\cap X_{\pi(j)}=\emptyset$ for all $i\le k$, $j>k$. This implies a bound on the diameter of the union of the cluster
\[
\diam\left(\bigcup_{i=1}^n X_i\right)=\diam(X)\le n r,
\]
with $r$ the range of the Hamiltonian which upper bounds $\diam(X_i)$ for all $i$.
In particular, for $\diam(X)=m$, we have at least $n=m/r$ contributing terms $h_{X_i}$.

If we assume that $\beta<\beta^*$, we observe that the above bounds hold for ${\Phi'}_X^A$ defined in terms of a rescaled interaction $h'_X=h_X \beta^*/\beta$ and conversely, since all the positive contributions in 
Equation~\eqref{eq:orderBound} are of order at least $m/r$ in the $h_X$, the original interaction fulfills 
\begin{align}
\sum_{X, \ \diam(X)=m}\|\Phi^A_X\| &\le\sum_{X, \ \diam(X)=m}\sum_{l=1}^\infty \sum_{\substack{(\gamma_1,\ldots,\gamma_l)\in\mathbb P^l\\ \gamma_1\lor\ldots\lor\gamma_l=X}}|\phi(\gamma_1,\ldots,\gamma_l)|\Pi_{j=1}^l\|\textbf{w}_\beta(\gamma_j)\|\\
&\le
\left(\frac{\beta}{\beta^*}\right)^{m/r}\sum_{X, \ \diam(X)=m}\sum_{l=1}^\infty \sum_{\substack{(\gamma_1,\ldots,\gamma_l)\in\mathbb P^l\\ \gamma_1\lor\ldots\lor\gamma_l=X}}|\phi(\gamma_1,\ldots,\gamma_l)|\Pi_{j=1}^l\|\textbf{w'}_\beta(\gamma_j)\|\\
&\le|\Lambda|\left(\frac{\beta}{\beta^*}\right)^{m/r}
\end{align}
with $\textbf{w'}$ defined in terms of $h'_X$, and so by choosing an appropriate $\ell=\mathcal{O}(\log(\Lambda/\eps))$ we conclude
\begin{align*}
\|\hefftot{A}-\heffexact{A}\|&\le \left\|\sum_{X, \ \diam(X)\ge\ell} \Phi^A_X\right\|\\
&\le\sum_{m\ge\ell}|\Lambda|\left(\frac{\beta}{\beta^*}\right)^{m/r}\\
&\le |\Lambda|\frac{(\beta/\beta^*)^{\ell/r}}{1-\beta/\beta^*} \, ,
\end{align*}
and thereby the proof of existence of an exponentially decaying effective interaction by (i).
\end{proof}

\section{Numerical analysis and circuit implementation of the rotated Petz recovery map}\label{sec:AppDetailsRotPetz}
We give details of a classical algorithm that outputs the gate decomposition of a quantum circuit $\eps$-approximately implementing the rotated Petz recovery map $\mathcal R_{\sigma_{BC},\Tr_C}^t$ for a given marginal $\sigma_{BC}$, which is provided as a classical input.
We assume a lower bound on its smallest eigenvalue denoted by $a$, and a number of input qubits $n$.
Note that the approach presented here might be far from optimal and only aims to provide the (quasi-)polynomial time guarantee for the quantum Gibbs sampling algorithm in the main text.
The techniques we use are standard in numerical analysis and we focus on the dependence of the scheme on the value $a$ but refrain from giving the details of connecting the error bounds of individual steps.
An exponential dependence on $n$, however, is to be expected and sufficient since $n$ is only (poly-)logarithmically large in the system size and inverse error.

We consider the integral
\[
\int_\mathbb R \beta_0(t) (\sigma_{BC}^{1/2-it}\sigma_B^{it-1/2})\otimes(\sigma_B^{-it-1/2}\sigma_{BC}^{1/2+it})dt
\]
representing the channel as a matrix.
The operator norm of the product of matrices is bounded by $1/a$ and due to the exponential decay of $\beta_0(t)\le\frac\pi2\exp(\pm\pi t)$, we obtain an $\eps$ approximation of the integral by truncating it to $[-t',t']$, $t'=\mathcal{O}(\log(1/a\eps))$.

Secondly, we need to approximate the matrix function $X\mapsto X^{\pm1/2\pm it}$ using its Taylor series and consider its convergence on $[a,1]$ inside which all eigenvalues are contained.
Writing the Taylor series as
\[
x^{-1/2+it}=\sum_{k=0}^\infty c_n (x-1)^n
\]
we can use Cauchy's integral formula for any path $\gamma$
\[
|c_n|=\left|\frac1{2\pi i}\int_\gamma \frac{x^{-1/2+it}}{(x-1)^{n+1}}\right|\le \sup_{x\textrm{ s.t. }|x-1|\le 1-a/2} \frac{|x^{-1/2+it}|}{(1-a/2)^n}\le \sqrt{2/a}\frac{e^{t'}}{(1-a/2)^n}
\]
by choosing $\gamma$ as a circle around 1 with radius $1-a/2$.
Then, within a disk of radius $1-a$, the terms of the series are bounded in norm by $\sqrt{2/a}e^{t'}(\frac{1-a}{1-a/2})^n$ and so bounding the remainder in the Taylor series
\begin{align*}
\sqrt{2/a}e^{t'}\sum_{n=k}^\infty\left(\frac{1-a}{1-a/2}\right)^n&\le\sqrt{2/a}e^{t'}\sum_{n=k}^\infty\left(1-a/2\right)^n\\
&=\sqrt{2/a}e^{t'}\frac{(1-a/2)^k}{a/2}
\end{align*}

we find that we can obtain an $\eps$-approximation by truncation to order
\begin{align*}
k\ge\frac{\log(\sqrt{2/a}^3 e^{t'}\eps)}{\log(1-a/2)}=\mathcal{O}\left(\frac{\log(1/a\eps)}{a}\right)\,.
\end{align*}
Bounds of the same order for $x^{\pm1/2\pm it}$ follow equivalently.

Finally for the numerical integration we can use Gauss quadrature convergence results for analytic functions, see \cite[Theorem 4.5]{Trefethen2008}. In particular, after rescaling the integration variable by $1/t'$, the integrand as a function of $t$ has its closest poles at $\pm i/t'$ and the integral goes over $[-1,1]$.
We consider the ellipse with focii $\pm1$ and minor semiaxis of length $1/2t'$. Then, inside the ellipse, $\beta_0$ is bounded by a constant.
Consider the remaining term in the rescaled integral
\[
\left( \sigma_{BC}^{1/2-itt'}\sigma_B^{itt'-1/2}\right)\otimes \left(\sigma_B^{-itt'-1/2}\sigma_{BC}^{1/2+itt'} \right) \, .
\]
In the ellipse, where $\textrm{Im}(t)\le1/2t'$ these terms can be bounded by a constant $\mathcal{O}(\poly(1/a))$.
According to \cite[Theorem 4.5]{Trefethen2008}, if the function is analytic and bounded inside the ellipse in absolute value by $M$ and for a sum of minor and major semiaxis $\rho\ge1+1/2t'$, the error in the Gaussian quadrature with $n$ evaluation points is bounded by
\[
\frac{64M}{15(1-\rho^{-2})\rho^{2n+2}}
\]
which proves an $\eps$-error for some
\[
n=\frac{\log(\frac{64M}{15\eps(1-\rho^{-2})})-2}{2\log(1+1/2t')}=\mathcal{O}(\poly \log(1/a\eps))\,.
\]

Together this gives a numerical algorithm to compute the matrix representation of the channel. Using standard linear algbraic transformations this can be decomposed into a Kraus map and by the Stinespring dilation into a unitary acting on a doubled Hilbert space, i.e., using additional ancilla qubits.

To finally obtain the gate decomposition into 1 and 2-qubit gates from any universal gateset we refer to the Solovay-Kitaev algorithm \cite{Kitaev1997}, which prescribes an algorithm that, in time $\exp(\mathcal{O}(n))\poly\log(1/\eps)$, constructs a gate sequence of depth $\exp(\mathcal{O}(n))\poly\log(1/\eps)$, where $n$ is the number of qubits the channel acts on.

\section{Continuity bounds for the dependence on the marginals of the rotated Petz recovery map}\label{sec:AppPetzCont}
In this Appendix we prove a continuity bound for the rotated Petz recovery map.
We start with a Lemma proving continuity of the matrix function $x\mapsto x^{\pm1/2\pm it}$, i.e., the squareroot with rotation.
The rotation can be seen as the barrier to obtaining a stronger, dimension-independent bound: Dimension independent bounds on the Fréchet derivative of operator monotone functions exist \cite[Chapter 10]{bhatia1997}, but do not apply because of the imaginary exponent here.
\begin{lemma}\label{lem:contSqrtRot}
Let $a>0$ and $X,Y\in\mathbb C^{n\times n}$ be hermitian matrices such that $a\id\le X,Y\le\id$. Then,
\begin{equation*}
\|X^{\pm1/2\pm it}-Y^{\pm1/2\pm it}\|\le (1+2|t|)\sqrt{n}a^{-3/2}\|X-Y\|
\end{equation*}
\end{lemma}
\begin{proof}
We start by separating real and imaginary part of the norm difference.
The real and imaginary parts of $x^{\pm1/2\pm it}$ are given by
\begin{align*}
f_\pm(x)=\sqrt{x^{\pm1}}\sin(t\log(x))\\
g_\pm(x)=\sqrt{x^{\pm1}}\cos(t\log(x))
\end{align*}
And using real and imaginary part of the matrix function and a triangle inequality we can decompose as follows.
\begin{equation}\label{eq:complexDecom}
\|X^{\pm1/2\pm it}-Y^{\pm1/2\pm it}\|\le \|f_\pm(X)-f_\pm(Y)\|+\|g_\pm(X)-g_\pm(Y)\|
\end{equation}
Here the $\pm$ on the rhs corresponds to the sign of the real part. The formula is the same for either sign of the imaginary part.
We have the derivatives
\begin{align}
f_\pm'(x)&=\frac{1}{\sqrt{x^{2\mp1}}}\left(\pm\frac12\sin(t\log(x))+t\cos(t\log(x))\right)\label{eq:compRealDeriv1}\\
g_\pm'(x)&=\frac{1}{\sqrt{x^{2\mp1}}}\left(\pm\frac12\cos(t\log(x))-t\sin(t\log(x))\right)\label{eq:compRealDeriv2}
\end{align}
all of which can be uniformly bounded by a constant $L=(1/2+|t|)a^{-3/2}$ on the interval $[a,1]$, which makes $L$ also a Lipschitz constant for these functions.

Consider the definition of the Fr\'echet derivative for matrix functions and a Hermitian matrix $A=U^\dagger\Lambda U$ that is diagonalizable with $\Lambda$ diagonal \cite[Section V.3]{bhatia1997}.
We define the matrix valued function
\begin{equation*}
\left(f^{[1]}(\Lambda)\right)_{i,j}=\begin{cases*}
\frac{f(\lambda_i)-f(\lambda_j)}{\lambda_i-\lambda_j} & if $\lambda_i\neq\lambda_j$\\
f'(\lambda_i) &\textrm{else}
\end{cases*}
\end{equation*}
for a diagonal matrix $\Lambda$ with diagonal $\lambda_1,\ldots,\lambda_n$.
Then the Fr\'echet derivative is given by \cite[Corollary V.3.2]{bhatia1997}
\begin{equation*}
Df(A)(H)=U\left(f^{[1]}(\Lambda)\circ(U^\dagger H U)\right)U^\dagger
\end{equation*}
where $\circ$ denotes the Hadamard product.
Using an interpolating path from $X$ to $Y$ as $X(t)=(1-t)X+tY$ we can write
\begin{align*}
\|f(X)-f(Y)\|&=\left\|f(X)-\int_0^1 Df(X(t))(Y-X)dt-f(X)\right\|\\
&\le\sup_{t\in[0,1]}\left\|f^{[1]}(\Lambda)\circ(U^\dagger(Y-X)U)\right\|\\
&\le\sqrt{n}L \|X-Y\|
\end{align*}
where $L$ is a Lipschitz function for $f$ valid on $[a,1]$ which includes the spectra of $X$, $Y$ and thereby $X(t)$.
Note that we are using that $L$ is a uniform bound on the entries of $f^{[1]}$.

Combining the inequality~\eqref{eq:complexDecom} with the bounds on derivatives from \eqref{eq:compRealDeriv1}-\eqref{eq:compRealDeriv2} and choosing $C=2\sqrt{n}L$ we conclude the result.
\end{proof}
With this at hand, we can continue to bound the operator to channel norm continuity of the rotated Petz map and prove Proposition~\ref{prop:rotPetzCont}.

\begin{proof}[Proof of Proposition~\ref{prop:rotPetzCont}]
We start with the continuity bound for $\mathcal R^t_{\cdot,\Tr_C}$ which is $t$-dependent.
By multiple application of triangle and H\"older's inequality, we obtain
\begin{align*}
\left\|\mathcal R^t_{\rho_{BC},\Tr_C}(X)-\mathcal R^t_{\sigma_{BC},\Tr_C}(X)\right\|_1&\le\|\rho_{BC}^{1/2-it}\rho_B^{it-1/2}X\rho_B^{-it-1/2}\rho_{BC}^{1/2+it}\\
&\qquad-\sigma_{BC}^{1/2-it}\sigma_B^{it-1/2}X\sigma_B^{-it-1/2}\sigma_{BC}^{1/2+it}\|_1\\
&\le \|\rho_{BC}^{1/2-it}-\sigma_{BC}^{1/2-it}\|a^{-1} \|X\|_1\\
&\qquad+\|\rho_{B}^{it-1/2}-\sigma_{B}^{it-1/2}\|a^{-1/2} \|X\|_1\\
&\qquad+\|\rho_{B}^{-it-1/2}-\sigma_{B}^{-it-1/2}\|a^{-1/2} \|X\|_1\\
&\qquad+\|\rho_{BC}^{1/2-it}-\sigma_{BC}^{1/2-it}\|a^{-1} \|X\|_1\\
\end{align*}
and applying Lemma~\ref{lem:contSqrtRot} and using $\|\rho_B-\sigma_B\|\le d_C \|\rho_{BC}-\sigma_{BC}\|$ we get
\begin{align*}
&\left\|\mathcal R^t_{\rho_{BC},\Tr_C}(X)-\mathcal R^t_{\sigma_{BC},\Tr_C}(X)\right\|_1\\
&\qquad\le\frac{2\|X\|_1}{a}\frac{(1+2|t|)\sqrt{d_{BC}}}{a^{3/2}}\|\rho_{BC}-\sigma_{BC}\|+\frac{2\|X\|_1}{a^{1/2}}\frac{(1+2|t|)\sqrt{d_B}}{a^{3/2}}d_C\|\rho_{BC}-\sigma_{BC}\|.
\end{align*}
Note that this bound is the same when defining the map on an enlarged Hilbert space so we obtain the diamond norm bound by dividing by $\|X\|_1$ and upper bounding constants
\begin{equation*}
\left\|\mathcal R^t_{\rho_{BC},\Tr_C}-\mathcal R^t_{\sigma_{BC}, \Tr_C}\right\|_\diamondsuit\le\frac{4(1+2|t|)d_{BC}}{a^{5/2}} \|\rho_{BC}-\sigma_{BC}\|.
\end{equation*}

We conclude by considering the integral definition of the recovery map and applying the triangle inequality 
\begin{align*}
\left\|\mathcal R^{t/2}_{\rho_{BC},\Tr_C}-\mathcal R^{t/2}_{\sigma_{BC},\Tr_C}\right\|_\diamondsuit&\le
\int_{\mathbb R}\beta_0(t)\left\|\mathcal R^{t/2}_{\rho_{BC},\Tr_C}-\mathcal R^{t/2}_{\sigma_{BC},\Tr_C}\right\|_\diamondsuit dt\\&\le \frac{8 d_{BC}}{a^{5/2}}\int_0^\infty \frac\pi2\frac{1+t}{\cosh(\pi t)+1}dt\|\rho_{BC}-\sigma_{BC}\|\\
&=\frac{4d_{BC}}{a^{5/2}}\left(1+\frac{\log(4)}{\pi}\right)\|\rho_{BC}-\sigma_{BC}\|
\end{align*}
and $4(1+\log(4)/\pi)<6$.
\end{proof}

\end{document}